\documentclass{article}

\usepackage{natbib}

\usepackage[margin=1in]{geometry}

\usepackage[utf8]{inputenc} 
\usepackage[T1]{fontenc}    
\usepackage{hyperref}       
\usepackage{url}            
\usepackage{booktabs}       
\usepackage{amsfonts}       
\usepackage{nicefrac}       
\usepackage{microtype}      

\usepackage{amsmath,amssymb,bbm,stmaryrd,mathrsfs,url,amsthm}
\newtheorem{theorem}{Theorem}
\newtheorem{lemma}{Lemma}
\usepackage{pstricks,framed}
\usepackage{amstext}
\usepackage{array,xcolor}
\usepackage{color,soul}
\usepackage{hyperref}
\usepackage{blkarray}
\usepackage{amsmath}
\usepackage{subcaption}
\usepackage{algorithm,mdframed}
\usepackage{algorithmic}
\usepackage{amsmath}
\usepackage[all,cmtip]{xy}
\usepackage{chemfig}
\usepackage{booktabs,makecell}
\usepackage{diagbox}
\usepackage[T1]{fontenc}
\usepackage[utf8]{inputenc}
\usepackage{easybmat}
\usepackage{multirow,bigdelim}
\usepackage{makeidx}
\usepackage{overpic}
\makeindex
\usepackage{tikz,siunitx}
\usetikzlibrary{arrows}
\usetikzlibrary{matrix,positioning,calc}
\usetikzlibrary{decorations.pathmorphing}
\tikzset{snake it/.style={-stealth,
	decoration={snake, 
			amplitude = 1.5mm,
			segment length = 2.5mm,
			post length=2.9mm},decorate}}
\newcommand{\change}{}

\newcommand{\R}{\mathbb{R}}

\renewcommand{\Re}[1]{\operatorname{Re}\left\{#1\right\}}

\newcommand{\vct}[1]{\boldsymbol{#1}}
\newcommand{\mtx}[1]{\boldsymbol{#1}}

\newcommand{\trace}{\operatorname{trace}}

\DeclareMathOperator*{\minimize}{\text{minimize}}

\newcommand{\va}{\vct{a}}
\newcommand{\vb}{\vct{b}}
\newcommand{\vc}{\vct{c}}

\newcommand{\ve}{\vct{e}}

\newcommand{\vh}{\vct{h}}

\newcommand{\vm}{\vct{m}}

\newcommand{\vs}{\vct{s}}

\newcommand{\vw}{\vct{w}}
\newcommand{\vx}{\vct{x}}
\newcommand{\vy}{\vct{y}}
\newcommand{\vz}{\vct{z}}

\newcommand{\vxi}{\vct{\xi}}
\newcommand{\veta}{\vct{\eta}}

\newcommand{\mA}{\mtx{A}}
\newcommand{\mB}{\mtx{B}}
\newcommand{\mC}{\mtx{C}}

\newcommand{\mX}{\mtx{X}}

\newcommand{\yl}{y_\ell}
\newcommand{\wl}{w_\ell}
\newcommand{\xl}{x_\ell}
\newcommand{\wol}{w^\natural_\ell}
\newcommand{\xol}{x^\natural_\ell}
\newcommand{\bl}{\vb_\ell}
\newcommand{\bltilde}{\tilde{\vb}_\ell}
\newcommand{\bktilde}{\tilde{\vb}_k}

\newcommand{\cl}{\vc_\ell}
\newcommand{\cltilde}{\tilde{\vc}_\ell}
\newcommand{\cktilde}{\tilde{\vc}_k}

\newcommand{\mnat}{\vm^\natural}
\newcommand{\hnat}{\vh^\natural}
\newcommand{\hhatnat}{\hat{\vh}^\natural}
\newcommand{\mhatnat}{\hat{\vm}^\natural}
\newcommand{\hnatt}{\vh^{\natural \intercal}}
\newcommand{\mnatt}{\vm^{\natural \intercal}}

\newcommand{\dm}{\delta \vm}
\renewcommand{\dh}{\delta \vh}
\newcommand{\dmtilde}{\widetilde{\delta \vm}}
\newcommand{\dhtilde}{\widetilde{\delta \vh}}
\DeclareMathOperator{\sign}{sign}

\def\l{\ell}

\newcommand{\ho}{\vh^\natural}

\newcommand{\mo}{\vm^\natural}
\newcommand{\mc}{\vm^\circ}
\newcommand{\hc}{\vh^\circ}
\newcommand{\mt}{\vm^\intercal}
\newcommand{\mot}{\vm^{\natural \intercal}}
\newcommand{\wo}{\vw^\natural}
\newcommand{\xo}{\vx^{\natural}}
\newcommand{\homot}{\ho\mot}
\newcommand{\blt}{\bl^\intercal}
\newcommand{\clt}{\cl^\intercal}
\newcommand{\clone}{c_{\ell1}}
\newcommand{\blone}{b_{\ell1}}
\newcommand{\dmt}{\dm^\intercal}
\newcommand{\cltildet}{\cltilde^\intercal}
\newcommand{\bltildet}{\bltilde^\intercal}
\newcommand{\cktildet}{\cktilde^\intercal}
\newcommand{\bktildet}{\bktilde^\intercal}
\newcommand{\PP}{\mathbb{P}}

\newcommand{\Rhnat}{\mathbf{R}_{\hnat}}
\newcommand{\Rhnatt}{\mathbf{R}_{\hnat}^\intercal}
\newcommand{\Rmnat}{\mathbf{R}_{\mnat}}
\newcommand{\Rmnatt}{\mathbf{R}_{\mnat}^\intercal}
\newcommand{\eoneeonet}{\ve_1 \ve_1^\intercal}

\title{BranchHull: Convex Bilinear Inversion from the Entrywise Product of Signals with Known Signs}

\author{
  Alireza Aghasi\thanks{
  aaghasi@gsu.edu, J. Mack Robinson College of Business, GSU}, Ali Ahmed\thanks{ali.ahmed@itu.edu.pk, Department of Electrical Engineering, ITU, Lahore}, Paul Hand\thanks{p.hand@northeastern.edu, Department of Mathematics and College of Computer and Information Science, Northeastern University}\ \ and Babhru Joshi\thanks{babhru.joshi@rice.edu, Department of Computational and Applied Mathematics, Rice University}
}

\begin{document}

\maketitle
\begin{abstract}
	We consider the bilinear inverse problem of recovering two vectors, $\vx$ and $\vw$, in $\R^L$ from their entrywise product. For the case where the vectors have known signs and belong to known subspaces, we introduce the convex program BranchHull, which is posed in the natural parameter space that does not require an approximate solution or initialization in order to be stated or solved. { Under the structural assumptions that {$\vx$ and $\vw$ are members of} known $K$ and $N$ dimensional random subspaces, we present a recovery guarantee for the noiseless case and a noisy case. In the noiseless case, we prove that the BranchHull recovers $\vx$ and $\vw$ up to the inherent scaling ambiguity with high probability when $L\ {\change \gg\ 2(K+N)}$. The analysis provides a precise upper bound on the coefficient for the sample complexity. In a noisy case, we show that with high probability the BranchHull is robust to small dense noise when $L = \Omega(K+N)$.} BranchHull is motivated by the sweep distortion removal task in dielectric imaging, where one of the signals is a nonnegative reflectivity, and the other signal lives in a known wavelet subspace.  Additional potential applications are blind deconvolution and self-calibration.
\end{abstract}

\section{Introduction}

This paper considers a bilinear inverse problem (BIP): recover vectors $\vx$ and $\vw$ from the observation $\vy = \mathcal{A}(\vx,\vw)$, where $\mathcal{A}$ is a bilinear operator. BIPs have been extensively studied in signal processing and data science literature, and comprise of fundamental problems such as blind deconvolution/demodulation [\cite{ahmed2012blind, stockham1975blind, kundur1996blind,aghasi2016sweep}], phase retrieval [\cite{fienup1982phase}], dictionary learning [\cite{tosic2011dictionary}], matrix factorization [\cite{hoyer2004non, lee2001algorithms}], and self-calibration [\cite{ling2015self}].
   Optimization problems involving bilinear terms and constraints also arise in other contexts, such as blending problems in chemical engineering [\cite{castro2015tightening}]. 

A significant challenge of BIPs is the ambiguity of solutions.  For example, if $(\vx^\natural,\vw^\natural)$ is a solution to a BIP, then so is $(c\vx^\natural,c^{-1}\vw^\natural)$ for any nonzero $c\in\mathbb{R}$.  Other ambiguities may also arise, including the shift ambiguity in blind deconvolution, the permutation ambiguity in dictionary learning, and the ambiguity up to multiplication by an invertible matrix in matrix factorization.  These ambiguities are challenging because they cause the set of solutions to be nonconvex. 

We will consider the fundamental bilinear inverse problem of recovering two $L$ dimensional vectors $\vw$ and $\vx$ from the observations $\vy = \vw \circ \vx$, where $\circ$ denotes the entry-wise product of vectors. This is immediately recognized as the calibration problem, where one is only able to measure a signal $\vx$ modulo unknown multiplicative gains $\vw$. A self-calibration algorithm aims to figure out the gains $\vw$ and the signal $\vx$ jointly from $\vy$.  The circular convolution also becomes pointwise multiplication in the Fourier domain, allowing us to reduce the important blind deconvolution problem in signal processing and wireless communications to a complex case of the above bilinear form.

In addition to the challenges of general BIPs, the BIP above is difficult because the solutions are nonunique without further structural assumptions.  For example $(\wo, \xo)$ and $(1, \wo \circ \xo)$ are both consistent with the entrywise products $\vy = \wo \circ \xo$.  While multiple structural assumptions are reasonable, we will consider the case where $\wo$ and $\xo$ belong to known subspaces $\mB$ and $\mC$, as in \cite{ahmed2012blind}.  In addition, we also require $\vw^\natural$ and $\vx^\natural$ to be real and of known signs. The method can be extended to complex vectors in the case of known complex phases.

The known sign information in the real case is justified in imaging applications, where we want to recover image pixels (always non-negative) from occlusions caused by unknown multiplicative masks {\change [\cite{Chen2006variation}]}. A stylized application of this setup also arises in the wireless communications. A source encodes a message as a series of positive magnitude shifts on tones at frequencies $f_1, f_2,\ldots, f_L$. These real valued and positive $\vx = [x(f_1), x(f_2), \ldots, x(f_L)]^\top$ are transmitted over a linear-time invariant channel, where $x(f_\ell)$ are weighted by the frequency response of the channel $w(f_\ell)$ (in general complex valued), and in the ideal noiseless case, the receiver ends up observing $y(f_\ell) = x(f_\ell)\cdot w(f_\ell)$. The real part of the complex-valued measurements $\Re{y(f_\ell)}= x(f_\ell) \cdot \Re{w(f_\ell)}$ are simply the pointwise product of two unknown real numbers with known signs. In addition, in this application, the vectors $\vx$, and $\vw$ naturally live in low-dimensional subspaces; for details, see \cite{ahmed2012blind} and \cite {ahmed2016leveraging}. 

The assumptions of sign and subspace measurements are strongly motivated by the sweep-distortion removal problem in dielectric imaging [\cite{aghasi2016sweep}].  In this problem, a dielectric is imaged, and the pointwise product of an electromagnetic pulse and the reflectivity pattern is observed.  The {\change signal's } nonnegativity follows from {\change nonnegativity} of the material's reflectivity, and the pulse belongs to a subspace defined by dominant wavelet coefficients of the image.

We consider the following bilinear inverse problem in the presence of multiplicative noise given by the vector $\boldsymbol{1} +\vxi$:
\begin{alignat}{2}
&\text{Let: } &&\wo \in \mB\subset \R^L, \ \xo \in \mC \subset \R^L, \vxi \in \R^L, \ \vs = \sign(\wo) \notag\\
& && \vy = \wo \circ \xo \circ (\boldsymbol{1}+\vxi), \label{signal-recovery-problem}\\
&\text{Given: } &&\vy,\vs,\mB,\mC \notag \\ 
&\text{Find: } &&\wo, \xo \text{ up to the scaling ambiguity}  \notag
\end{alignat}

One standard way to solve the BIP above\footnote{As stated, this approach ignores the sign information.} is to convexify it by lifting.  More specifically, the bilinear inverse problem can be recast as a linear matrix recovery problem with the structural constraint that the recovered matrix is rank one. With $\wol = \blt \ho$ and $\xol = \clt \mo$ for $\l = 1, \ldots, L$,  the underlying linear operator is given by $y_\ell = \wol \xol = \blt \ho \mot \cl  = \langle \bl \clt, \ho \mot \rangle=\mathcal{A}_\ell(\ho\mot)$, 
and the formal recovery framework is {\change to} find the $\mX$ of minimal rank that is consistent with $\mathcal{A}(\mX) = \vy$. 
By relaxing the rank objective to the nuclear norm of $\mX$, this optimization problem becomes a semidefinite program. The results in \cite{ahmed2012blind}, which apply to the complex case, show that when $\bl$ and $\cl$ are Fourier and Gaussian vectors, respectively, this semidefinite program  succeeds in recovering the rank-1 matrix $\ho\vm^{\natural \top}$ with high probability, whenever $K+N\lesssim L/\log^3L$.  Unfortunately, directly optimizing a lifted problem is prohibitively computationally expensive, as the lifted semidefinite program is posed on a space of dimensionality $K \times N$, which is much larger than the $K+N$ dimensionality of the natural parameter space.

To address the intractability of lifted methods,  a recent theme of research has been to solve quadratic and bilinear recovery problems in the natural parameter space using alternating minimization and gradient descent algorithms [\cite{netrapalli2013phase, sun2016geometric}].  These algorithms include the Wirtinger Flow (WF) and its variants for phase retrieval {\change [\cite{candes2014phase, CC15, wang2016solving}]}. A Wirtinger gradient descent method was recently introduced for blind deconvolution in \cite{li2016rapid}.  In the case that $\bl$ are deterministic complex matrices that satisfy an incoherence property and that $\cl$ are Gaussian vectors, this nonconvex method succeeds at recovering $\ho$ and $\mo$ up to the scale ambiguity with high probability when $K+N \lesssim L/\log^2 L$.   While WF based methods enjoy rigorous recovery guarantees under optimal or nearly optimal sample complexity with suitable measurement models,  the proofs of these results are long and technical.  Also, because of the nonconvexity of the problem, the convergence of a {\change gradient descent algorithm} to the global minimum usually relies on an appropriate initialization [\cite{tu2015low, CC15, li2016rapid}].

The approach we will present in this paper will combine strengths of both of these approaches.  Specifically, we introduce a convex formulation in the natural parameter space for the bilinear inverse problem of recovering two real vectors from their entrywise product, provided that the vectors live in known subspaces and have known signs.  This convex formulation is called BranchHull and does not involve an initialization or approximate solution in order to be posed or solved. BranchHull is based on the following idea:  The bilinear measurements $\xl \wl = \yl$ establish that $(\xl, \wl)$ is on one of two branches of a hyperbola in $\R^2$.  Information on $\sign(\wl)$ identifies the appropriate branch.  The convex formulation is then formed by relaxing this nonconvex branch of a hyperbola to its convex hull, as shown in Figure~\ref{fig:hyperbola-convex-hull}. 
 We consider the case where the two vectors live in random subspaces of $\R^L$ with dimensions $K$ and $N$. Under this assumption on $\xo$ and $\wo$, with noise $\vxi$ that does not change the sign of the measurements, we establish that the Euclidean recovery error 
is bounded by the $\l_\infty$ norm of the noise. This result holds with high probability for $K+N \lesssim L$. In the noiseless case, we provide an explicit lower bound on the recovery probability that is nonzero when {\change $L > 2(K+N)-3$}.

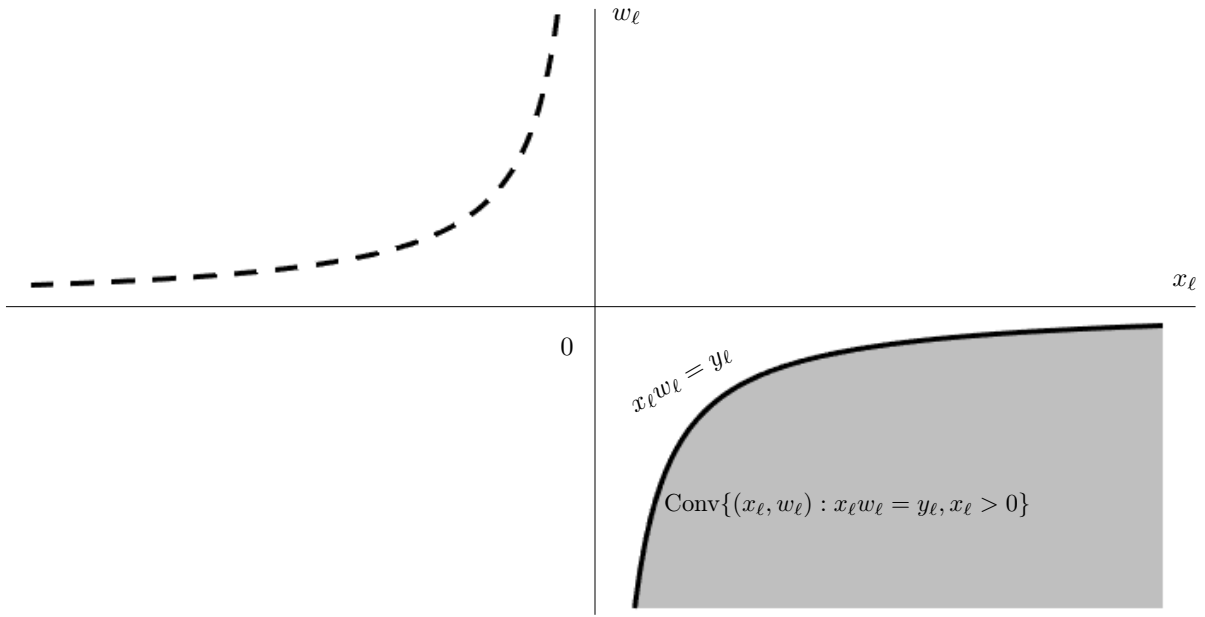
\begin{figure}[t]\centering 
	\begin{overpic}[ width=0.93\textwidth,height=0.5\textwidth,tics=1]{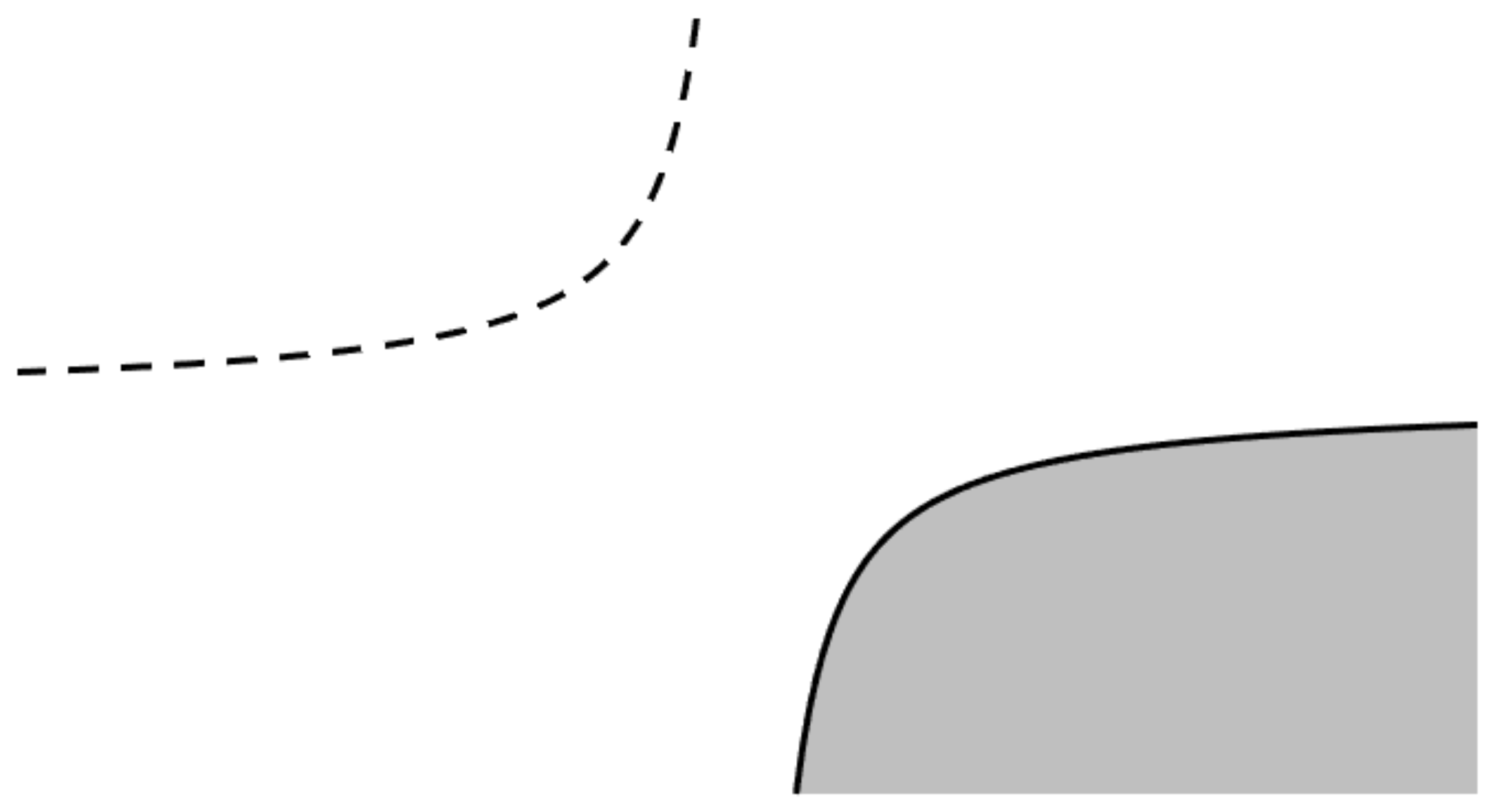}
		\linethickness{0.3pt}
		\put(-1,27.2){\color{black}\line(1,0){103}}
		\put(50,.5){\color{black}\line(0,1){52.5}}
		\put(47,23){$0$}
		\put(51.5,52){$w_\ell$}
		\put(100,29){$x_\ell$}
		\put(53,18){\rotatebox{30}{$x_\ell w_\ell =y_\ell $}}
		\put(56,9.5){\rotatebox{0}{\scalebox{.93}{$\mbox{Conv}\!\left\{(x_\ell,w_\ell): x_\ell w_\ell =y_\ell, x_\ell>0\right\}$}}}
	\end{overpic}
	\caption{\small Given the bilinear measurement $\xl \wl = \yl$, the point $(\xl, \wl)$ is on a two-branch hyperbola, as depicted by the dashed and solid lines.  Further information on the sign of $\wl$ identifies which branch of the hyperbola the point is on (the solid line).  The convex formulation in this paper replaces the relevant branch of the hyperbola with its convex hull (the shaded region).  }
	\label{fig:hyperbola-convex-hull}
\end{figure}

\subsection{Problem Formulation}
We consider the bilinear inverse problem of recovering two vectors from their entrywise product.  That is, let $\wo,\xo,\vxi \in \R^L$, and let $\vy = \wo \circ \xo\circ ({\bf 1}+\vxi)$, where $\vxi$ corresponds to noise.  From $\vy$, we attempt to find $\wo$ and $\xo$ up to the scaling ambiguity $(c\wo, \frac{1}{c} \xo)$.  To make the problem well posed, we consider the case where $\wo$ and $\xo$ belong to known subspaces $\mB$ and $\mC$ of $\R^L$.  We further consider the case where the signs of the entries of $\wo$, and hence those of $\xo$, are known.  Let $\vs = \sign(\wo)$.   This bilinear inversion problem is stated in \eqref{signal-recovery-problem}.

Ideally, we could resolve the scaling ambiguity and find $(\wo, \xo)$ such that $\|\wo\|_2 = \|\xo\|_2$  by solving the following program:
\begin{align*}
\minimize_{\vw\in \mB, \ \vx\in \mC} \ \|\vw\|_2^2+\|\vx\|_2^2  \text{ subject to }  &\wl \xl = \yl\\[-.5em]
&s_\l \wl \geq 0, ~ \l = 1, \ldots, L. 
\end{align*}
This program is nonconvex, but it admits the following convex relaxation:
\begin{align*}
\minimize_{\vw \in \mB, \ \vx \in \mC} \ \|\vw\|_2^2+\|\vx\|_2^2  \text{ subject to } &\sign(\yl) \wl \xl \geq |\yl| \\[-.5em]
&s_\l \wl \geq 0, ~ \l = 1, \ldots, L. 
\end{align*}
Note that for fixed $\ell$, the feasible set $\{(\wl, \xl) \mid \sign(\yl)\wl \xl \geq | \yl|, s_\l \wl \geq 0 \}$ is the convex hull of $\{(\wl, \xl) \mid \wl \xl =  \yl, s_\l \wl \geq 0 \}$.

We consider this problem when written in the natural parameter space.  {\change Despite the abuse of notation,} Let $\mB \in \R^{L \times K}$ be a matrix that spans the $K$ dimensional subspace $\mB$.  Similarly, let $\mC \in \R^{L \times N}$ be a matrix that spans the $N$ dimensional subspace $\mC$.  Let $(\ho, \mo) \in \R^K \times \R^N$.  Let $\wo = \mB\ho$ and $\xo = \mC \mo$.       We can write $\wl = \blt \vh$, $\xl = \clt \vm$, and $\yl = \langle \bl \clt, \homot \rangle$, where $\blt$ is the $\ell$th row of $\mB$ and $\clt$ is the $\ell$th row of $\mC$.  The recovery task is now to find $(\ho, \mo)$ by the convex program called BranchHull
\begin{align}
\minimize_{\vh \in \R^K, \ \vm \in \R^N} \ \|\vh\|_2^2+\|\vm\|_2^2  \text{ subject to } &\sign(\yl) \blt \vh \cdot \clt \vm \geq |\yl| \tag{BH} \label{bh}\\
&s_\l  \cdot \blt \vh \geq 0, ~ \l = 1, \ldots, L. \notag
\end{align}
This program is convex because for any fixed $\l$, the points consistent with both the first and second constraints is a convex set. This program has $K+N$ variables, $L$ linear inequality constraints, and $L$ nonlinear inequality constraints. Because the scaling $(c \wo, \frac{1}{c}\xo)$ is consistent with the constraints for positive $c$, the program will return a solution where $\|\vh\|_2 = \|\vm\|_2$.  Thus, if recovery is successful in the noiseless case, the optimal solution is $\Biggl(\ho \sqrt{\frac{\|\mo\|_2}{\|\ho\|_2}}, \mo \sqrt{\frac{\|\ho\|_2}{\|\mo\|_2}} \Biggr)$. We implement the same convex program \eqref{bh} in the noisy case if the noise $\vxi$ does not alter the sign of the measurement $\vy$. This occurs when $\xi_{\l} \geq -1$ for all $\l \in [L]$. For the case where noise alters sign or outlier case, see the discussion section for a modified program that is conjectured to tolerate sign change and significant outliers.

\subsection{Main Results}
In this paper, we consider the bilinear recovery problem \eqref{signal-recovery-problem},  where the subspaces given by $\mB$ and $\mC$ are random.  Specifically, we show that if $\mB$ and $\mC$ have i.i.d. Gaussian entries, then exact recovery of $(\ho, \mo)$ is possible in the noiseless case with nonzero probability when there are {\change at least 2 times} as many measurements as degrees of freedom. 

\begin{theorem}[Noiseless Case] \label{thm:sd}
	Fix $(\ho, \mo) \in \R^{K} \times \R^{N}$ such that $\ho \neq 0$ and $\mo \neq 0$.  Let $\mB\in\R^{L\times K},\mC \in\R^{L\times N}$ have i.i.d. $\mathcal{N}(0, 1)$ entries and $\vxi = 0$.  Then $\Biggl(\ho \sqrt{\frac{\|\mo\|_2}{\|\ho\|_2}}, \mo \sqrt{\frac{\|\ho\|_2}{\|\mo\|_2}} \Biggr)$ is the unique solution to \eqref{bh} with probability at least $$
	1 - \exp\Biggl( - \frac{\bigl[L - {\change (2N + 2K - 3)}\bigr]^2}{2(L-1)} \Biggr), $$
	provided that {\change $L > 2N + 2K - 3$}.
\end{theorem}
This theorem provides an explicit lower bound on the recovery probability by the convex program \eqref{bh}.  If $L > {\change 2N+2K-3}$, there is a nonzero probability of successful recovery.  By taking $L \geq \tilde{C} (N+K)$, the probability of failure becomes at most $e^{-\tilde{c} L}$, for universal constants $\tilde{C}$ and $\tilde{c}$.  The scaling of $L$ in terms of $N+K$ is information theoretically optimal up to a constant factor.  The proof of Theorem~\ref{thm:sd} follows from estimating the probability of covering a sphere by random  hemispheres chosen from a nonuniform distribution.

Now we will state a result that the convex program \eqref{bh} is robust to small dense noise. Let
\begin{align}
		 \epsilon &= \|\vxi\|_{\infty} \label{epsilon}
\end{align}
represent the noise level. In particular, we present a recovery theorem for $\epsilon \leq 1$. Under this assumption on the noise, we show that if the matrices $\mB$ and $\mC$ have i.i.d. Gaussian entries and there are $O(K+N)$ measurements, then the minimizer of \eqref{bh} is close to $\left( \ho \sqrt{\frac{\|\mo\|_2}{\|\ho\|_2}},\mo\sqrt{\frac{\|\ho\|_2}{\|\mo\|_2}}\right)$ with high probability.   

\begin{theorem}[Noisy Case]\label{thm:noise}
	Fix $(\ho,\mo) \in \mathbb{R}^K \times \mathbb{R}^N$ such that $\ho \neq 0$ and $\mo \neq 0$. Let $\mB\in \mathbb{R}^{L\times K}$, $\mC\in \mathbb{R}^{L\times N}$ have i.i.d. $\mathcal{N}(0,1)$ entries. Let $\epsilon$ be as defined in \eqref{epsilon}. Let $\vy \in \mathbb{R}^L$ contain measurements that satisfy \eqref{signal-recovery-problem} with $\epsilon \in [0,1]$. If $L \geq C(K+N)$ then the unique minimizer $(\vh^*, \vm^*)$ of the BranchHull program \eqref{bh} satisfies
	\begin{equation}
		\left(\left\|\vh^* - \ho\sqrt{\frac{\|\mo\|_2}{\|\ho\|_2}}\right\|_2^2+\left\|\vm^* - \mo \sqrt{\frac{\|\ho\|_2}{\|\mo\|_2}}\right\|_2^2\right)^\frac{1}{2}\nonumber \leq 4\sqrt{\epsilon}\sqrt{\|\ho\|_{2}\|\mo\|_{2}}
	\end{equation} 
	with probability at least $1-e^{-cL}$. Here, $C$ and $c$ are absolute constants.
\end{theorem}
{\change In Theorem \ref{thm:noise}, the $\ell_2$ recovery error depends on the noise level as the square root of $\epsilon$. We suspect that this square root dependence in the power of $\epsilon$ is an artifact of the the proof technique and numerical simulations presented in Figure \ref{vary_noise_level} suggests the recovery errors, for small noise, behaves linearly in $\epsilon$.}

\subsection{Discussion}

The BranchHull formulation is a novel convex relaxation for the bilinear recovery from the entrywise product of  vectors with known signs, and it enjoys a recovery guarantee when those vectors belong to random real subspaces of appropriate dimensions.   The formulation is nothing more than finding which point of a convex set is closest to the origin.   Geometrically, exact recovery is possible by $\ell_2$-norm minimization because the feasible set of $(\vh, \vm)$ has a `pointy' ridge that corresponds to the fundamental scaling ambiguity, as illustrated in Figure~\ref{fig:pointy-ridge}.

\begin{figure}[t]\centering 
	\begin{overpic}[scale = 0.1,trim= 0 400 0 0]{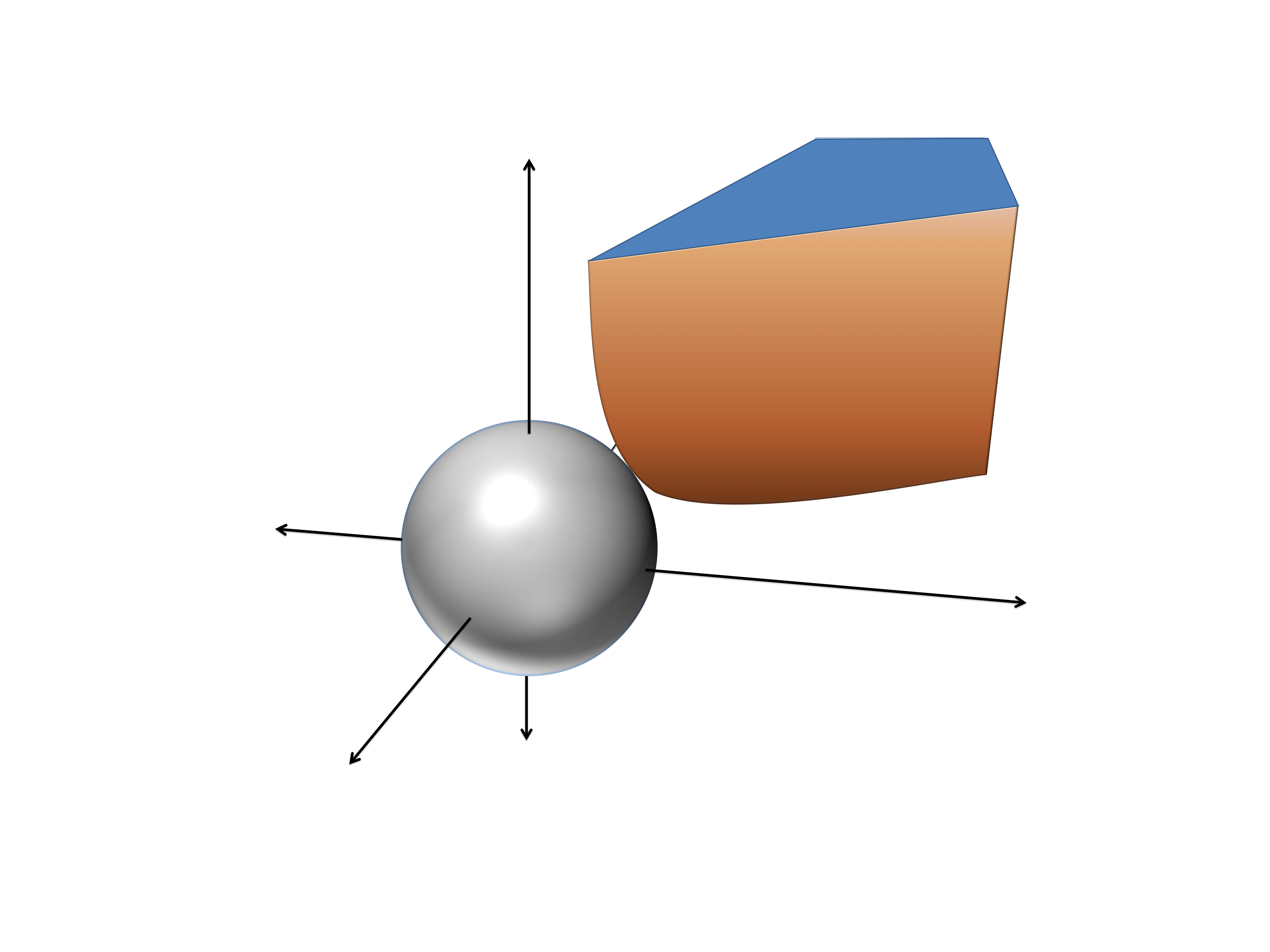}
		\linethickness{0.3pt}
		\put(20,15){$h_1$}
		\put(78,10){$h_2$}
		\put(36.75,48){$m_1$}
	\end{overpic}
	\caption{ \small An illustration of the geometry of BranchHull in the case where $\vh \in \R^2$ and $\vm \in \R^1$.  
		The feasible set of BranchHull has a shape similar to the solid in the top right.  The ridge of this set corresponds to the fundamental scaling ambiguity of the bilinear recovery problem.   The solution to BranchHull is given by the smallest scaling of the unit ball that intersects the feasible set.  The minimizer is exactly on this ridge because the ridge is `pointy.'  }
	\label{fig:pointy-ridge}
\end{figure}

A related formulation to BranchHull was recently introduced for the phase retrieval problem.  This formulation, called PhaseMax, is a linear program and was independently discovered by \cite{bahmani2016phase} and \cite{goldstein2016phasemax}.  PhaseMax enjoys a rigorous recovery guarantee under a random data model.  Existing recovery proofs are based on statistical learning theory [\cite{bahmani2016phase}], geometric probability [\cite{goldstein2016phasemax}], and elementary probabilistic concentration arguments [\cite{hand2016elementary}].  As with Wirtinger Flow, successful recovery of PhaseMax with optimal sample complexity has been proven when an appropriate initialization is known.  Unlike Wirtinger Flow, the initialization is used in PhaseMax's objective function, as opposed to its algorithmic implementation.  In both PhaseMax and Wirtinger Flow, an approximate solution or initialization is needed to state or solve the optimization problem. {\change We note that BranchHull does not require an anchor vector or initialization to be stated or solved. As a trade off, BranchHull instead assumes the sign information of the signal is known. }


The idea of convex relaxations in the natural parameter space for bilinear problems is not new.  For example, in nonlinear programming (NLP) or mixed integer nonlinear programming (MINLP) problems with bilinear constraints and specified variable bounds, a McCormick relaxation [\cite{mccormick1976computability}] replaces bilinear terms with four linear inequality constraints that define a convex quadrilateral that contains the hyperbola of feasible points within the variable bounds.  Tighter convex relaxations are possible [\cite{castro2015tightening}], such as by using the hyperbola itself as an inequality constraint [\cite{quesada1995global}].  These relaxations have been studied as part of branch and bound approaches to NLPs and MINLPs. Under certain conditions and branching rules [\cite{horst2013global}] these approaches can find a global minimizer; however, the branching results in many convex programs that need to be solved, and it may result in exponential time complexity.  In contrast, the present paper considers only the single convex program, BranchHull, achieved by the natural convex relaxation of bilinear constraints with only sign information.  This work establishes conditions --- in particular, subspace conditions --- under which exact recovery by an efficient convex program can be rigorously established.

This work motivates several interesting and important extensions. Most immediately, BranchHull can be extended when the phases of complex vectors are known.  Because of applications in signal processing and communications, it is also important to extend  the theory for BranchHull to include deterministic subspaces, such as the span of partial Fourier matrices. This paper shows that BranchHull is robust to noise that does not change  the sign of any measurement. Tolerance to a general noise model, including outliers, should be established for BranchHull or a variation with slack variables, such as in \cite{hand2016corruption}.  Noise tolerance in the case with sign change is particularly important because even one measurement with an incorrect sign can substantially alter the shape of the feasible set. For this general noise and outlier case, we propose the Robust BranchHull program
\begin{align}
	\minimize_{\vh \in \R^k, \vm \in \R^N,\ve \in \R^L} \|\vh\|_{2}^{2}+\|\vm\|_{2}^{2}+\lambda\|\ve\|_{1}  \text{ subject to }  &\sign(y_{\l})(\clt \vm + e_{\l})\blt \vh \geq |y_{\l}|, \tag{RBH} \label{rbh}\\
	 & s_{\l} \cdot \blt \vh \geq 0, ~ \l = 1,\dots, L \notag,
\end{align}	
which shifts the feasibility set to include the target signal while penalizing against shift. In the outlier case, the $\l_1$ penalty promotes sparsity of slack variable $\ve$, which is desired. We leave detailed empirical and theoretical analysis of \eqref{rbh} to future work. It would also be interesting to develop convex relaxations in the natural parameter space that do not use sign information.  Further, extensions to more general bilinear recovery problems are of significant interest.  All of these directions are left for future publications.


\subsection{Organization of the paper}
The remainder of the paper is organized as follows. In Section \ref{notation}, we present notations used throughout the paper. In Section \ref{sd proof}, we present the proof of Theorem \ref{thm:sd}. In Section \ref{noise proof}, we present the proof of Theorem \ref{thm:noise}. In Section \ref{simulation} we observe the performance of BranchHull on synthetic data.

\subsection{Notation}\label{notation}

Vectors and matrices are written with boldface, while scalars and entries of vectors are written in plain font.  For example, $\clone$ is the first entry of $\cl$.  We write $\boldsymbol{1}$ as the vector of all ones with dimensionality  appropriate for the context.  Let $[L] = \{1, 2, \ldots, L\}$.  Let $\ve_i$ be the $i$th standard basis element.  We write $K+N \lesssim L$ to mean that there exists a constant $C$ such that $K+N \leq C L$.  Given a vector in $\vx \in \R^N$, let $\tilde{\vx} \in \R^{N-1}$  be the subvector formed by all but the first coefficient of $\vx$.  Let $\mathbb{S}^{n-1}$ be the unit sphere in $\R^n$.  For matrices $\mA, \mB$, let $\langle \mA,\mB \rangle = \trace(\mB^\intercal \mA)$ be the Hilbert-Schmidt inner product of $\mA$ with $\mB$.  For a set $S$, let $\text{Conv}(S)$ be its convex hull.  Let $\Re{z}$ be the real part of a complex $z$. 

\section{Technical Proofs}\label{proofs}
 In this section we provide proofs of Theorems \ref{thm:sd} and \ref{thm:noise}. These proofs use a sphere covering type argument which is based on the idea that $m$ random directions sampled from a symmetric distribution will cover the unit sphere $\mathbb{S}^{n-1}$ with high probability when $m = \Omega(n)$. {\change Another} paper that uses this technique is the PhaseMax paper {\change by} \cite{goldstein2016phasemax}.

\subsection{Proof of Theorem \ref{thm:sd}}\label{sd proof}

{\change 
We will first show that BranchHull program \eqref{bh} is a convex program.
\begin{lemma}
If $\vy \in \mathbb{R}^L$ such that $y\neq 0$, $\vs\in \{\pm 1\}^L,\ \mB \in \mathbb{R}^{L\times K}$ and $\mC \in \mathbb{R}^{L\times N}$ then the BranchHull program \eqref{bh} is a convex program.	
\end{lemma}
\begin{proof}
	As the objective function is convex, we consider the constraints of \eqref{bh}. For a fixed $\ell$, let $S_{\ell}=\{(\vh,\vm) \in \mathbb{R}^K\times\mathbb{R}^N: \sign(\yl)\blt\vh\cdot\clt\vm\geq|\yl|,\ s_\ell \cdot\blt\vh\geq0 \}$, $S_{\l,1} = \{(x,w)\in \mathbb{R}^2: \sign(\yl)xw\geq |\yl|,\ s_\l w\geq0\}$ and $S_{\l,2} = \{(\vh,\vm)\in\mathbb{R}^K\times\mathbb{R}^N:(\blt\vh,\clt\vm)\in S_{\l,1}\}$. To show $S_\l$ is a convex set, it is sufficient to show that the sets $S_{\l,1}$ and $S_{\l,2}$ are convex. 
	
	We first show that the set $S_{\l,1}$ is convex. Let $P = \{w\in \mathbb{R} : s_\l w \geq 0\}$ and consider the function $f:P \rightarrow \mathbb{R}$ such that $f(w) = \frac{\yl}{w}$. Note that if $\sign(\yl)s_\l\geq 0$ then $f$ is a convex function and $S_{\l,1}$ is the epigraph of $f$. Similarly, if $\sign(\yl)s_\l\leq 0$ then $f$ is a concave function and $S_{\l,1}$ is the subgraph of $f$. In both cases, $S_{\l,1}$ is a convex set because the epigraph of a convex function and the subgraph of a concave function are convex.
	
	Lastly, $S_{\l,2}$ is convex because the inverse image of a convex set of a linear map is convex. So, $S_\l$ is also a convex set. Since the intersection of any number of convex sets is convex, we have that the constraint of \eqref{bh} is a convex set. Thus, BranchHull program \eqref{bh} is a convex program.
	\end{proof}
	}
To prove Theorem \ref{thm:sd}, we will show that  $(\ho, \mo)$ is the unique minimizer of an optimization with a larger feasible set defined by linear constraints.  

\begin{lemma}  \label{lem:sd-to-sdlc}
	If $(\ho, \mo)$ is the unique solution to
	\begin{align}
	\minimize_{\vh\in \R^K, \vm\in\R^N} \ \|\vh\|_2^2+\|\vm\|_2^2  \text{ subject to } \yl \langle \bl\clt, \vh \mot + \ho \mt \rangle \geq 2 \yl^2,& \label{bilc} \\[-.5em] \ \ell = 1, \ldots, L,& \notag
	\end{align}
	then $(\ho,\mo)$ is the unique solution to \eqref{bh}.
\end{lemma}

\begin{proof}[Proof of Lemma \ref{lem:sd-to-sdlc}]
	It suffices to show that the feasible set of \eqref{bilc} contains the feasible set of \eqref{bh}.  We may rewrite \eqref{bh} as 
	\begin{align*}
	\minimize_{\vh\in\R^K, \vm \in \R^N} \ \|\vh\|_2^2+\|\vm\|_2^2  \text{ subject to } &\yl \blt \vh \cdot \clt \vm  \geq \yl^2\\[-.5em]
	&s_\l \cdot \blt \vh \geq 0, ~ \l = 1, \ldots, L.
	\end{align*}
	We now use the fact that a convex set with a smooth boundary is contained in a halfspace defined by the tangent hyperplane at any point on the boundary of the set. Consider the point $(\wol, \xol) \in \R^2$, and observe that
	\begin{align}
	\left \{ (\wl, \xl) \in \R^2 \bigg | \begin{matrix}\yl \wl \xl \geq \yl^2 \\ \sign(\wl) = s_\l \end{matrix} \right \} \subseteq \left \{ (\wl, \xl) \in \R^2 \bigg | \begin{pmatrix}
	\yl \xol \\ \yl \wol \end{pmatrix} \cdot \begin{pmatrix}\wl - \wol \\ \xl - \xol \end{pmatrix} \geq 0 \right \}.
	\end{align}
	Plugging in $\wl = \blt \vh$ and $\xl = \clt \vm$, we have that any feasible $(\vh, \vm)$ satisfies $$\yl \clt \mo \blt \vh + \yl \blt \ho \clt \vm \geq 2 \yl^2, \quad \ell = 1, \ldots, L,$$
	which implies $\yl \langle \bl\clt, \vh \mot + \ho \mt \rangle \geq 2 \yl^2$ for all $\ell$.
\end{proof}

We now show that $(\ho, \mo)$ is the unique solution to the optimization problem \eqref{bilc} if the unit sphere in $\R^{N+K-2}$ is covered by $L$ hemispheres given in terms of $\bl$ and $\cl$.   Write $\bl = (\blone, \bltilde)$, where $\bltilde$ contains all but the first {\change element} of $\bl$.  Similarly, write $\cl = (\clone, \cltilde)$.

\begin{lemma} \label{lem:sdlc-sufficient-condition}
	Let $\ho = \ve_1$ and $\mo = \ve_1$.  The unique solution to \eqref{bilc} is $(\ho, \mo)$ if for all $(\dmtilde, \dhtilde) \in \R^{N-1} \times \R^{K-1}$ there exists an $\ell \in [L]$ such that $\blone \neq 0$, $\clone \neq 0$, and
	\begin{align} \Biggl\langle \frac{\cltilde}{\clone}, \dmtilde \Biggr \rangle + \Biggl \langle \frac{\bltilde}{\blone}, \dhtilde \Biggr \rangle \leq 0.
	\end{align}
\end{lemma}

\begin{proof}[Proof of Lemma \ref{lem:sdlc-sufficient-condition}]
	Because the feasible set of \eqref{bilc} is closed and convex, and because a closed convex set has a unique point closest to the origin, \eqref{bilc} has a unique minimizer.
	
	Consider a feasible point $(\ho + \dh, \mo + \dm)$. To prove that $(\ho,\mo)$ is a minimizer of \eqref{bilc}, it suffices to show
	\begin{align*}
	\langle \bl \clt, \ho\mot\rangle\langle \bl \clt, \dh\mot + \ho \dmt \rangle \geq 0 \ \forall \  \ell \Rightarrow \langle \mo, \dm \rangle + \langle \ho, \dh \rangle \geq 0.
	\end{align*}
	Plugging in $\ho = \ve_1$ and $\mo = \ve_1$, it suffices to show
	\begin{align*}
	\blone \clone \bigl[ \blone \clone(\dm_1 + \dh_1) + \blone \cltildet \dmtilde + \clone \bltildet \dhtilde \bigr]\geq 0 \ \forall \ \ell \Rightarrow \delta m_1 + \delta h_1 \geq 0
	\end{align*} 
	Dividing by $\blone^2 \clone^2$, it suffices to show
	\begin{align*}
	\delta m_1 + \delta h_1 + \biggl\langle \frac{\cltilde}{\clone}, \dmtilde \biggr\rangle + \biggl\langle \frac{\bltilde}{\blone}, \dhtilde \biggr\rangle \geq 0  \ \forall \ell \text{ s.t. } \blone \neq 0 \text{ and } \clone \neq 0 &\\\Rightarrow \delta m_1 + \delta h_1 \geq 0&.
	\end{align*}
	To prove this, it suffices to prove
	\begin{align*}
	\forall (\dhtilde, \dmtilde) \in \R^{N-1} \times \R^{K-1}, \exists \ \ell \text{ s.t. } &\blone \neq 0, \clone \neq 0,~ \text{and } ~ \biggl\langle \frac{\cltilde}{\clone}, \dmtilde \biggr\rangle + \biggl\langle \frac{\bltilde}{\blone}, \dhtilde \biggr\rangle \leq 0.
	\end{align*}
\end{proof}

For a given vector $\va$, we will call $\{\boldsymbol{\delta}\in\mathbb{S}^{n-1}:\langle\va,\boldsymbol{\delta}\rangle\geq 0\}$ the hemisphere centered at $\va$. We now provide a lower bound to the probability of covering the unit sphere by hemispheres centered at  $m$ random directions under a nonuniform probability distribution that is symmetric to negation.  This lemma is an immediate generalization of  Lemma 2 in  \cite{goldstein2016phasemax}, with a nearly identical proof.

\begin{lemma} \label{lem:sphere-covering}
	Choose $m$ independent random vectors $\{\va_i\}_{i=1}^m$ in $\mathbb{S}^{n-1}$ from a (possibly nonuniform) distribution that is symmetric with respect to negation, and is such that all subsets of size $n$ are linearly independent with probability $1$.  Then, the hemispheres centered at $\{\va_i\}_{i=1}^m$ cover the whole sphere with probability
	$$
	1 - \frac{1}{2^{m-1}} \sum_{k=0}^{n-1} {m-1 \choose k}.
	$$
	This value is the probability of flipping at least $n$ heads among $m-1$ tosses.
\end{lemma}

\begin{proof}[Proof of Lemma \ref{lem:sphere-covering}]
	Classical arguments  in sphere covering [\cite{wendel1962problem}] show\footnote{This article credits \cite{schlafli1953gesammelte} for the proof argument.} the following: If $m$ hyperplanes containing the origin are such that the normal vectors to any subset of $n$ hyperplanes are linearly independent, then the complement of the union of these hyperplanes  is partitioned into  
	$$
	r(n,m) = 2 \sum_{k=0}^{n-1} {m-1 \choose k}
	$$
	connected regions. In each of these regions,  every point lies on the same side of each hyperplane.  Alternatively put, each region  corresponds to a unique assignment of a side of each hyperplane.  For a fixed set of $m$ hyperplanes, if the half space on either side  of each hyperplane is selected by independent tosses of a fair coin,  then with {\change probability given in the lemma statement}, there will be no nontrivial intersection {\change of all} these half spaces. 
	
	By the assumption that the distribution of $\va_i$ is symmetric with respect to negation, we have that for any $\vz \in \mathbb{S}^{n-1}$, the conditional distribution of $\va_i$ given $\va_i \in \{ \pm \vz\}$ is uniform over the two elements $\pm \vz$.  By independence, for any fixed  $\{\vz_i\}_{i=1}^m \in (\mathbb{S}^{n-1})^m$, the distribution of $\{\va_i\}$ conditioned on the event $\{\va_i = \pm \vz_i\}$ is uniform over the $2^{m}$ possibilities.  Thus, conditioned on this event, {\change the probability that the sphere is covered is that of the lemma statement}.  Integrating over all possible $\{\vz_i\}$, the lemma follows.
\end{proof}

Our last technical lemma provides an explicit lower bound to the probability of the sphere covering given in Lemma \ref{lem:sdlc-sufficient-condition}.

\begin{lemma} \label{lem:probability-estimate}
	With probability at least $1 - \exp\Biggl( - \frac{\bigl[(L-1) - 2{\change(N+K-2)}\bigr]^2}{2(L-1)} \Biggr)$, we have that
	\begin{align}
	\forall (\dhtilde, \dmtilde) \in \R^{N-1} \times \R^{K-1}, \exists \ \ell \text{ such that }  \biggl\langle \frac{\cltilde}{\clone}, \dmtilde \biggr\rangle + \biggl\langle \frac{\bltilde}{\blone}, \dhtilde \biggr\rangle \leq 0. \label{opt-cond}
	\end{align}
\end{lemma}
\begin{proof}[Proof of Lemma \ref{lem:probability-estimate}]
	To show \eqref{opt-cond}, we must show that the $L$ hemispheres (of the unit sphere in $\R^{N-1} \times \R^{K-1}$) centered at $(-\frac{\cltilde}{\clone}, -\frac{\bltilde}{\blone})$ cover the entire sphere.  As the distribution of $(\frac{\cltilde}{\clone}, \frac{\bltilde}{\blone})$ is invariant to negation, and as any $n$ samples from this distribution are linearly independent with probability 1, Lemma \ref{lem:sphere-covering} gives that the probability that \eqref{opt-cond} holds is at least the probability of flipping at least {\change $N+K-2$} heads among $L-1$ tosses of a fair coin.
	
	We now bound the probability of getting at least $n$ heads among $m$ fair coin tosses.  Let $X$ be the number of heads in $m$ tosses.  By Hoeffding's inequality for Bernoulli random variables  [\cite{wasserman2013all}], for any $t \geq 0$, $$\PP\Bigl(X-\frac{m}{2}>-mt \Bigr)\geq 1-e^{-2mt^2}.$$ By selecting $t = \frac{1}{2} - \frac{n}{m}$, we get that when $n \leq m/2$,
	$$
	\PP(\text{at least $n$ heads among $m$ tosses}) \geq 1 - e^{-2m(\frac{1}{2} - \frac{n}{m})^2 } = 1 - e^{-\frac{(m-2n)^2}{2m}}.
	$$
	
	The lemma follows by plugging in $m = L-1$ and {\change $n = N + K - 2$} into the above probability estimate.
\end{proof}
Now, we may prove the theorem.
\begin{proof}[{\bf Proof of Theorem \ref{thm:sd}}]
	
	Without loss of generality, let $\|\hnat\|_2 = \|\mnat\|_2$.  This is possible because for any $\bl$ and $\cl$,  we have that  $\Biggl(\ho \sqrt{\frac{\|\mo\|_2}{\|\ho\|_2}}, \mo \sqrt{\frac{\|\ho\|_2}{\|\mo\|_2}} \Biggr)$ and $(\hnat, \mnat)$ give equal values of $\yl = \langle \bl \clt, \ho\mot\rangle$.
	
	Further, without loss of generality, let $\| \hnat\|_2 = \| \mnat\|_2 = 1$.  This is possible because the scaling
	$$
	\hat{\vh} = \frac{\vh}{\| \hnat\|_2}, \quad \hat{\vm} = \frac{\vm}{\| \mnat\|_2}, \quad \hat{\hnat} = \frac{\hnat}{\| \hnat\|_2}, \quad \hat{\mnat} = \frac{\mnat}{\| \mnat\|_2},
	$$
	turns \eqref{bh} into
	\begin{align}
	\minimize_{\vh \in \R^K, \ \vm \in \R^N} \ &\|\hnat\|_2^2 \|\hat{\vh}\|_2^2+ \|\mnat\|_2^2 \|\hat{\vm}\|_2^2 \label{norm1} \\ \text{ subject to } &\langle \bl \clt, \hhatnat \hat{\mot} \rangle \langle \bl \clt, \hat{\vh} \hat{\mt} \rangle \geq \langle \bl \clt, \hhatnat \hat{\mot} \rangle^2
	\notag\\
	&s_\l  \cdot \blt \hat{\vh} \geq 0, ~ \l = 1, \ldots, L. \notag
	\end{align}
	
	Further, without loss of generality we may take $\hnat = \ve_1$ and $\mnat = \ve_1$. To see this is possible, let $\Rhnat$ and $\Rmnat$ be rotation matrices that map $\hnat$ and $\mnat$ to $\ve_1$, respectively.   Letting $\bar{\vh} = \Rhnat \vh$, $\bar{\vm} = \Rmnat \vm$, and $\bar{s}_\l = \sign(\blt \Rhnatt \ve_1)$, problem \eqref{bh} can be written
	\begin{align}
	\minimize_{\bar{\vh} \in \R^K, \ \bar{\vm} \in \R^N} \ & \|\Rhnat \bar{\vh}\|_2^2+ \|\Rmnat \bar{\vm}\|_2^2 \\ \text{ subject to } &\langle \Rhnat \bl \clt \Rmnatt, \eoneeonet \rangle \langle \Rhnat \bl \clt \Rmnatt, \bar{\vh} \bar{\mt} \rangle \geq \langle \Rhnat \bl \clt \Rmnatt, \eoneeonet \rangle^2
	\notag\\
	&\bar{s}_\l  \cdot \blt \Rhnatt \bar{\vh} \geq 0, ~ \l = 1, \ldots, L. \notag
	\end{align}
	As $\ell_2$ norms are invariant to rotation and as $\Rhnat \bl$ and $\Rmnat \cl$ have independent $\mathcal{N}(0, 1)$ entries, we may take $(\hnat, \mnat) = (\ve_1, \ve_1)$. 
	
	Let $E$ be the event that \eqref{opt-cond} holds.  By Lemma~\ref{lem:probability-estimate}, $$\PP(E) \geq 1 - \exp\Biggl( - \frac{\bigl[(L-1) - 2{\change (N+K-2)}\bigr]^2}{2(L-1)} \Biggr).$$  By Lemma~\ref{lem:sdlc-sufficient-condition},  on $E$, $(\hnat, \mnat)$ is the unique solution to \eqref{bilc}.  By Lemma~\ref{lem:sd-to-sdlc}, on $E$, $(\hnat, \mnat)$ is the unique solution to \eqref{bh}.
\end{proof}

\subsection{Proof of Theorem \ref{thm:noise}}\label{noise proof}
We will now prove that BranchHull is robust to small dense noise which does not alter the sign of the measurements. The sign of the {\change measurements remain} unchanged when $\xi_\l \geq -1$ for all $\l \in [L]$. We first only consider measurements $\yl$ with noise $\xi_\l$ that satisfy 
\begin{equation}\label{one-sided}
	\xi_\ell \in[-1,0] 	
\end{equation}
for all $\ell \in [L]$. If all the measurements satisfy condition \eqref{one-sided}, we say the noise is "one-sided". Note that the noise is one-sided if the convex hull of the branch of the hyperbola corresponding to the noisy measurement contain the hyperbola corresponding to the noiseless measurement, for all measurements. We first establish a recovery result for measurements with one-sided noise and show that for measurements that contain $\xi_\l >0$, the problem can be transformed to a related scaled problem whose corresponding measurements contain one-sided noise.

For the remainder of the paper, let $\hat{y_{\l}} = \blt \hnat \clt \mnat$. Lemma \ref{good noise thm} {\change shows} that if the measurements contain one-sided noise, the recovery error using BranchHull program \eqref{bh} is bounded by $\|\vxi\|_\infty$.
\begin{lemma}\label{good noise thm}
	Let $\ho = e_1$ and $\mo = e_1$. Let $\mB\in \mathbb{R}^{L\times K}$, $\mC\in \mathbb{R}^{L\times N}$ and $y_{\ell}$ satisfy \eqref{signal-recovery-problem} such that the noise is one-sided as in \eqref{one-sided}. Let $\epsilon = \|\vxi\|_\infty$. The minimizer $(\vh^*,\vm^*)$  of the BranchHull program \eqref{bh} is unique and satisfies
	\[\left\|\vh^* - \ho\right\|_2^2+\left\|\vm^* - \mo\right\|_2^2\leq 4(1-\sqrt{1- \epsilon})\]
	if for all $(\dhtilde,\dmtilde)\in \mathbb{R}^{K-1}\times\mathbb{R}^{N-1}$, there exists $\ell$, $k \in [L]$ such that
	\begin{align}
		&\sign(\blone)\bltildet\dhtilde\leq0 \text{ and } \sign(\clone)\cltildet\dmtilde\leq 0,   \label{cond1}\\ 
		&\sign(b_{k1})\bktildet\dhtilde\geq 0 \text{ and } \sign(c_{k1})\cktildet\dmtilde\leq 0. \label{cond2}
	\end{align}	
\end{lemma}
\begin{proof}
	First note that the minimizer of BranchHull program \eqref{bh} is unique because the feasible set is closed and convex and a closed convex set has a unique point closest to the origin. We now prove the remainder of lemma \ref{good noise thm} by showing that any feasible perturbation from the candidate minimizer increases the objective value of the BranchHull program \eqref{bh}. 
	
	Assume the minimizer of \eqref{bh} is $(\ho+\dh, \mo + \dm)$. Note that $(\ho, \mo)$ is feasible in \eqref{bh} because the noise is one-sided. Comparing the objective values at  $(\ho+\dh, \mo + \dm)$ and $(\ho, \mo)$, we get
	\begin{equation}\label{compare obj}
		\|\dh\|^2_2 + \|\dm\|^2_2 \leq -2\left(\hnatt \dh + \mnatt \dm\right) = -2\left(\delta h_1 + \delta m_1\right).
	\end{equation}
	
	We now use the second feasibility condition $s_\ell \blt \vh\geq 0$ to show $\delta h_1 \geq -1$. Since $(\ho+\dh, \mo + \dm)$ is feasible, the following holds for all $\ell \in [L]$.
	\begin{align}
		& s_\ell \blt(\ho + \dh)\geq 0 \nonumber \\
		\Rightarrow & \sign(\blone)\left(\blone+\blone\delta h_1 + \bltildet\dhtilde\right)\geq 0 \label{sign constraint}\\
		\Rightarrow & |\blone|\delta h_1 \geq -|\blone| -\sign(\blone) \bltildet\dhtilde \nonumber\\
		\Rightarrow & \delta h_1 \geq -1 - \frac{\sign(\blone) \bltildet\dhtilde}{|\blone| \nonumber}\\
		\Rightarrow & \delta h_1 \geq -1, \label{dm bound}
	\end{align}
	where the first implication holds because $s_{\ell} = \sign(\blt\ho)$ and $\ho = e_1$ and the last implication holds because, by assumption \eqref{cond1}, there exists a $\l \in [L]$ such that $\sign(\blone)\bltildet\dhtilde\leq0$.
	
	We now use the first feasibility condition on $(\ho+\dh, \mo + \dm)$ to show that $\delta h_1 + \delta m_1$ is bounded from below. From the first feasibility condition, for all $\ell \in [L]$ we have
	\begin{align}
			&\sign(y_\ell)\blt(\ho + \dh)\clt(\mo + \dm)\geq |y_\ell| \nonumber\\
			\Rightarrow &\sign(y_\ell)(\blone + \blt\dh)(\clone + \clt\dm)\geq |y_\ell|\nonumber\\
			\Rightarrow & \sign(y_\ell)\left(\blone\clone + \blone \clt \dm + \blt\dh \clone + \blt\dh \clt\dm\right) \geq |y_\ell|\nonumber\\
			\Rightarrow & \sign(y_{\ell})\blone\clone(\delta h_1 + \delta m_1)+ \sign(y_\ell)\left(\blone\cltildet\dmtilde + \bltildet\dhtilde\clone + \blt\dh \clt\dm\right) \geq  |y_\ell| - \sign(y_\ell)\blone\clone\nonumber\\
			\Rightarrow &|\hat{y_\ell}|(\delta h_1 + \delta m_1) + \underbrace{\sign(y_\ell)\left(\blone\cltildet\dmtilde + \bltildet\dhtilde\clone + \blt\dh \clt\dm\right)}_{I} \geq |\hat{y_\ell}|\xi_\ell, \label{first feas 1}
		\end{align}
	where the first implication holds because $\ho = \ve_1$ and $\mo = \ve_1$ and the last implication holds because $\sign(y_\ell) = \sign(\hat{y_\ell})$ and $|y_\ell| = \sign(y_\ell)\blone\clone(1+\xi_\ell)$.
	
	We now show that term $I$ is less than $|\hat{y_\ell}|\delta h_1 \delta  m_1$ for some $\ell \in [L]$. Consider
		\begin{align}
			I & = \sign(y_\ell)\left(\blone\cltildet\dmtilde + \bltildet\dhtilde\clone + \blt\dh \clt\dm\right) \nonumber \\
			& = \sign(y_\ell)\left(\blone\cltildet\dmtilde + \bltildet\dhtilde\clone + \blone\clone\delta h_1 \delta m_1 + \blone\delta h_1 \cltildet\dmtilde+\bltildet\dhtilde\clone\delta m_1 + \bltildet\dhtilde\cltildet\dmtilde\right)\nonumber \\
			& = |\hat{y_\ell}|\delta h_1 \delta m_1 +\sign(\hat{y_\ell})(1+\delta m_1)\bltildet\dhtilde\clone + \sign(\hat{y_\ell})\cltildet\dmtilde\left(\blone +\blone\delta h_1 + \bltildet\dhtilde\right)\nonumber \\
			& = |\hat{y_\ell}|\delta h_1 \delta m_1 +\underbrace{(1+\delta m_1)|\clone|\sign(\blone)\bltildet\dhtilde + \sign(\clone)\cltildet\dmtilde\left|\blone +\blt{\dh}\right|}_{II},\nonumber 
		\end{align}
	where the third equality holds because $\sign(y_\ell)=\sign(\hat{y_\ell}) = \sign(\blone\clone)$ and the fourth equality holds thanks to \eqref{sign constraint}. Note that because of assumptions \eqref{cond1} and \eqref{cond2}, there exists a $\ell \in [L]$ such that $II \leq 0$. This is because if $(1+\delta m_1) \geq 0$, then we have $II \leq  0$ for $\ell$ that satisfy \eqref{cond1}. Similarly, if $(1+\delta m_1) < 0$, then we have $II \leq 0$ for $k$ that satisfy \eqref{cond2}. Thus, there exists an $\ell \in [L]$ such that 
	\begin{equation} \label{first feas 2}
		I = \sign(y_\ell)\left(\clone\cltildet\dmtilde + \bltildet\dhtilde\clone + \blt\dh \clt\dm\right) \leq |\hat{y_\ell}|\delta h_1 \delta m_1.
	\end{equation} 
	Combining \eqref{first feas 1} and \eqref{first feas 2}, we get there exist a $\ell$ such that 
	\begin{equation}
		\delta h_1 + \delta m_1 + \delta h_1 \delta m_1 \geq \frac{|\hat{y_\ell}|\xi_\ell}{|\hat{y_\ell}|}\geq \xi_\l \geq -\epsilon.
	\end{equation}
	The last inequality holds because $\epsilon = \|\vxi\|_\infty$. Lastly, $\delta h_1 + \delta m_1 \geq -2\left(1-\sqrt{1-\epsilon}\right)$ because for all $\epsilon \in [0,1]$, 
	\begin{equation}
		\begin{aligned}
			&\left\{(\delta h_1,\delta m_1)\in \mathbb{R}^2\big| \delta h_1 + \delta m_1 \geq -2\left(1-\sqrt{1-\epsilon}\right)\right\}\\
			&\supset \left\{(\delta h_1,\delta m_1)\in \mathbb{R}^2\big| \delta h_1 + \delta m_1  + \delta h_1 \delta m_1 \geq -\epsilon, \ \delta h_1 \geq -1 \right\}.
		\end{aligned}	
	\end{equation}
	Thus, combining \eqref{compare obj} with $\delta h_1 + \delta m_1 \geq -2\left(1-\sqrt{1-\epsilon}\right)$, we get the desired result $\left\|\dh\right\|_2^2+\left\|\dm\right\|_2^2\leq 4(1-\sqrt{1-\epsilon})$.	
\end{proof}

The next lemma shows that that if $\mB \in \mathbb{R}^{L \times K}$ and $\mC\in \mathbb{R}^{L \times N}$ contain i.i.d. $\mathcal{N}(0,1)$ entries, then \eqref{cond1} and \eqref{cond2} holds with high probability if $L = \Omega(K+N)$. 

\begin{lemma}\label{sphere covering}
Let $\mB \in \mathbb{R}^{L \times K}$ and $\mC\in \mathbb{R}^{L \times N}$ contain i.i.d. $\mathcal{N}(0,1)$ entries. If $L\geq C(K+N)$ then 
\[\min_{\vx \in \mathbb{S}^{K-1}, \vy \in \mathbb{S}^{N-1}} \sum_{\ell = 1}^{L} \mathbbm{1}_{\blt \vx\leq 0}\mathbbm{1}_{\clt \vy\leq 0} \geq 0.2L\]
with probability at least $1 - e^{-cL}$. Here, $C$ and $c$ are absolute constants.
\end{lemma}

\begin{proof}

Let $f(\vx,\vy) = \sum_{\ell = 1}^{L} \mathbbm{1}_{\blt \vx\leq 0}\mathbbm{1}_{\clt \vy\leq 0}$. We will consider a continuous relaxation of $f(\vx,\vy)$. Let 
\[
w(z) = \left\{\begin{array}{ll}
			1 &  z<-0.1\\
			-\frac{z}{0.1} &  -0.1\leq z\leq0\\
			0 &  z>0
		\end{array}\right.
\]
and $g(\vx,\vy) = \sum_{\l=1}^{L} w(\blt \vx)w(\clt \vy)$. Note that $f(\vx,\vy)\geq g(\vx,\vy)$ for all $(\vx,\vy)$. So, it is sufficient to show that with probability at least $1-e^{-cL}$,
\begin{equation}\label{cont sphere}
	\min_{x \in \mathbb{S}^{K-1}, y \in \mathbb{S}^{N-1}} \sum_{\l=1}^{L} w(\blt \vx)w(\clt \vy) \geq 0.2L .
\end{equation} 
if $L\geq C(K+N)$. 

Let $\beta_{\l}(\vx,\vy) = w(\blt \vx)w(\clt \vy)$. We first compute $\mathbb{E}[\beta_{\l}(\vx,\vy)]$ for a fixed $\vx \in \mathbb{S}^{K-1}$ and $\vy \in \mathbb{S}^{N-1}$. Without loss of generality, let $\vx = \ve_{1}$ and $\vy = \ve_1$.
\begin{align*}
\mathbb{E}[\beta_{\l}(\vx,\vy)] = &\mathbb{E}[w(\blone)w(\clone)] \\
= & \left(\mathbb{E}[w(\blone)]\right)^2\\
= &\left(\frac{1}{\sqrt{2\pi}}\int_{-\infty}^{-.1}e^{-\frac{s^2}{2}}ds + \frac{1}{\sqrt{2\pi}}\int_{-.1}^{0}\left(-\frac{s}{0.1}\right)e^{-\frac{s^2}{2}}ds \right)^2\\
\geq &\ 0.23,
\end{align*}
where the second inequality follows by independence of $\blone$ and $\clone$. So, for a fixed $(\vx,\vy) \in \mathbb{S}^{K-1}\times \mathbb{S}^{N-1}$, we have  $\mathbb{E}[g(\vx,\vy)] \geq 0.23L$. 

We will now show that for a fixed $(\vx,\vy) \in \mathbb{S}^{K-1}\times \mathbb{S}^{N-1}$, $g(\vx,\vy) \geq 0.22L$  with high probability. Fix $(\vx,\vy) \in \mathbb{S}^{K-1}\times \mathbb{S}^{N-1}$. Since $g(\vx,\vy)$ is bounded, $g(\vx,\vy)$ is sub-gaussian. Let $\alpha$ be the sub-gaussian norm of $\beta_{\l}$ after centering. Thus, by Hoeffding-type inequality (see Proposition 5.10 in \citet{vershynin10in}), $\mathbb{P}\left\{\left|g(\vx,\vy)-\mathbb{E}[g(\vx,\vy)]\right| \geq t\right\} \leq e\cdot e^{-\frac{ct^2}{\alpha^2L}}$, where $c>0$ is a absolute constant. So, $\mathbb{P}\left\{g(\vx,\vy)\leq 0.23L-\delta L\right\} \leq e\cdot e^{-\frac{c\delta^2L}{\alpha^2}}$. Pick $\delta = 0.01$, then for any fixed $(\vx,\vy) \in \mathbb{S}^{K-1}\times \mathbb{S}^{N-1}$, we have 
\[\mathbb{P}\left\{g(\vx,\vy)\leq 0.22L\right\} \leq e\cdot e^{-cL}\]
 for some $c>0$.

We will now show that for all $(\vx,\vy)$ in an $\epsilon$-net, $g(\vx,\vy) \geq 0.22L$  with high probability. Let $\mathcal{N}_{\epsilon}$ be an $\epsilon$-net of $\mathbb{S}^{K-1}\times \mathbb{S}^{N-1}$ such that $|\mathcal{N}_\epsilon |\leq (1+\frac{2\sqrt{2}}{\epsilon})^{K+N}$.  By lemma 5.2 in \cite{vershynin10in}, such an $\epsilon$-net exists. So 
\begin{equation}\label{e-net}
\mathbb{P}\left\{\min_{(\vx,\vy)\in \mathcal{N}_\epsilon} g(\vx)\geq 0.22L\right\} \geq 1- e\cdot e^{-cL +(N+K)\log(1+\frac{2\sqrt{2}}{\epsilon})}.
\end{equation}
If $L\geq \frac{2}{c}(1+\log(1+\frac{2\sqrt{2}}{\epsilon}))(K+N)$ then 
\begin{equation}
\mathbb{P}\left\{\min_{(\vx,\vy)\in \mathcal{N}_\epsilon} g(\vx,\vy)\geq 0.22L\right\} \geq 1- e\cdot e^{-\frac{cL}{2}}
\end{equation}

Lastly, we will show that for all $(\vx,\vy) \in \mathbb{S}^{K-1}\times \mathbb{S}^{N-1}$, $g(\vx,\vy) \geq 0.2L$  with high probability. We first show that $g(\vx,\vy)$ is $30\sqrt{2}L$-Lipschitz with high probability. This holds because if $(\vx_1,\vy_1), (\vx_2,\vy_2) \in \mathbb{R}^{K-1}\times \mathbb{R}^{N-1}$ then
\begin{align}
|g(\vx_1,\vy_1) - g(\vx_2,\vy_2)| \leq &\sum_{\l = 1}^{L}|w(\blt \vx_1)w(\clt \vy_1)-w(\blt \vx_2)w(\clt \vy_2)|\\
= &\sum_{\l = 1}^{L}\left|\left(w(\blt \vx_1)-w(\blt \vx_2)\right)w(\clt \vy_1) +\left(w(\clt \vy_1)-w(\clt \vy_2)\right)w(\blt \vx_2)\right|\\
\leq &\sum_{\l = 1}^{L}\left|\left(w(\blt \vx_1)-w(\blt \vx_2)\right)w(\clt \vy_1)\right| + \left|\left(w(\clt \vy_1)-w(\clt \vy_2)\right)w(\blt \vx_2)\right|\\
\leq &\sum_{\l = 1}^{L}\left|w(\blt \vx_1)-w(\blt \vx_2)\right| + \left|w(\clt \vy_1)-w(\clt \vy_2)\right|\\
\leq & 10\sum_{\l = 1}^{L}|\blt (\vx_1-\vx_2)| +10\sum_{\l = 1}^{L}|\clt (\vy_1-\vy_2)|\\	
\leq &10\sqrt{L} \cdot \sqrt{\sum_{\l=1}^{L}(\blt (\vx_1-\vx_2))^2} +10\sqrt{L} \cdot \sqrt{\sum_{\l=1}^{L}(\clt (\vy_1-\vy_2))^2}\\
= &10\sqrt{L}\left(\|\mB(\vx_1-\vx_2)\|_{2}+\|\mC(\vy_1-\vy_2)\|_{2}\right)\\
\leq &10\sqrt{L}\left(\|\mB\|\|(\vx_1-\vx_2)\|_{2}+\|\mC\|\|(\vy_1-\vy_2)\|_{2}\right),
\end{align}
where the fourth line follows because $|w(z)| \leq 1$ for all $z \in \mathbb{R}$, the fifth line follows because $w$ is 10-Lipschitz and the sixth line follows from Cauchy-Schwarz inequality. By Corollary 5.35 in \cite{vershynin10in}, there exists events $E_1$ and $E_2$ each with probability at least $1-2e^{-\frac{L}{2}}$, on which $\|\mB\|\leq 3\sqrt{L}$ and $\|\mC\| \leq 3\sqrt{L}$,	 respectively. So, on $E_{1}\cap E_{2}$ we have
\begin{align*}
	|g(\vx_1,\vy_1) - g(\vx_2,\vy_2)| & \leq 30L(\|\vx_1-\vx_2\|_2+\|\vy_1-\vy_2\|_2)\\
	&\leq  30\sqrt{2}L\|(\vx_1,\vy_1)-(\vx_2,\vy_2)\|_{2}.	
\end{align*}
Take $\epsilon = \frac{0.01}{30\sqrt{2}}$. For any $(\vx_1,\vy_1) \in \mathbb{S}^{K-1}\times \mathbb{S}^{N-1}$, pick $(\vx_2,\vy_2)\in \mathcal{N}_{\epsilon}$ such that $\|(\vx_1,\vy_1)-(\vx_2,\vy_2)\|_{2} \leq \epsilon$. On the event $E_{1}\cap E_{2}$ and the event given by \eqref{e-net}, we have that 
\begin{align}
	g(\vx_1,\vy_1)&\geq \min_{(\vx_2,\vy_2)\in\mathcal{N}_{\epsilon}}g(\vx_2,\vy_2)-30\sqrt{2}L\|(\vx_1,\vy_1)-(\vx_2,\vy_2)\|_{2}\\
		&\geq 0.2L.
\end{align} 
This occurs with the probability at least {\change $1-e^{-cL}$}, provided $L \geq \frac{2}{C}(1+\log(1+\frac{2\sqrt{2}}{\epsilon}))(K+N)$.	
\end{proof}

We now extend lemma \ref{good noise thm} to the case with arbitrary but non-zero $\ho$ and $\mo$. 
\begin{lemma}\label{good noise arbitrary}
Fix $\ho \in \mathbb{R}^K$ and $\mo \in \mathbb{R}^N$ such that $\ho \neq 0$ and $\mo \neq 0$. Let $\mB\in \mathbb{R}^{L\times K}$, $\mC\in \mathbb{R}^{L\times N}$ contain i.i.d $\mathcal{N}(0,1)$ entries and $y_{\ell}$ satisfy  \eqref{signal-recovery-problem} such that the noise is one-sided as in \eqref{one-sided}. Let $\epsilon = \|\vxi\|_\infty$. The minimizer $(\vh^*,\vm^*)$ of the BranchHull program \eqref{bh} is unique and if $L\geq C(K+N)$ then the minimizer satisfies
	\[\left\|\vh^* - \ho\sqrt{\frac{\|\mo\|_2}{\|\ho\|_2}}\right\|_2^2+\left\|\vm^* - \mo\sqrt{\frac{\|\ho\|_2}{\|\mo\|_2}}\right\|_2^2\leq 4(1-\sqrt{1-\epsilon}){\|\ho\|_2\|\mo\|_2}\]
	with probability at least $1-e^{-cL}	$. Here, $C$ and $c$ are absolute constants.
\end{lemma}
\begin{proof}
	Without loss of generality let $\|\ho\|_{2} = \|\mo\|_2$, which is possible because for any $b_{\l}$, $c_\l$ and $\xi_\l$, we have that $\left(\ho\sqrt{\frac{\|\mo \|_2}{\|\ho\|_2}},\mo\sqrt{\frac{\|\ho\|_2}{\|\mo\|_2}}\right)$ and $(\ho,\mo)$ give equal values of $y_{\l} = \blt \ho \clt \mo(1+\xi_\l)$. 
	Further, without loss of  generality, we may take $\|\ho\|_{2} = \|\mo\|_2 = 1$. This is possible because of a similar line of argument as \eqref{norm1}.
	 
	Further, without loss of generality we may take $\ho = \ve_1$ and $\mo = \ve_1$. To see this is possible, let $R_{\hhatnat}$ and $R_{\mhatnat}$ be rotation matrices that map $\hhatnat$ and $\mhatnat$ to $\ve_1$, respectively. Letting $\bar{\vh} = R_{\hhatnat}\hat{\vh}, \ \bar{\vm} = R_{\mhatnat}\hat{\vm}$, and $\bar{s_\l} = \sign(\blt R_{\hhatnat}^\intercal \ve_1)$, BranchHull can be written as
	\begin{equation}\label{rotation invariant}
		\begin{aligned}
			&\min_{\bar{\vh}\in \mathbb{R}^{K}, \bar{\vm}\in \mathbb{R}^N} \|\ho\|_2^2\|R_{\hhatnat}^\intercal \bar{\vh}\|_2^2 + \|\mo\|_2^2\|R_{\mhatnat}^\intercal \bar{\vm}\|_2^2\\
			& \text{s.t. } \sign\left((R_{\hhatnat}\bl)^\intercal \ve_1(R_{\mhatnat}\cl)^\intercal \ve_1 \right)(R_{\hhatnat}\bl)^\intercal \bar{\vh}(R_{\mhatnat}\cl)^\intercal \bar{\vm}  \geq \left|(R_{\hhatnat}\bl)^\intercal \ve_1(R_{\mhatnat}\cl)^\intercal \ve_1(1+\xi_\l)\right|\\
			& \hspace{7mm} \bar{s_\l} \cdot \blt R_{\hhatnat}^\intercal \bar{\vh}\geq 0,\ \l \in [L].
		\end{aligned} 
	\end{equation}
	As $\ell_{2}$ norms are invariant to rotation and $R_{\hhatnat}\bl$ and $R_{\mhatnat}\cl$ have independent $\mathcal{N}(0,1)$ entries, we may take $(\ho, \mo) = (\ve_1, \ve_1)$.
	
	By Lemma \ref{good noise thm}, the minimizer $(\bar{\vh}^*,\bar{\vm}^*)$ of \eqref{rotation invariant} is unique and  satisfies
	\begin{equation}\label{error bound transformed}
		\left\|\bar{\vh}^* - \ve_1\right\|_2^2+\left\|\bar{\vm}^* - \ve_1\right\|_2^2\leq 4(1-\sqrt{1-\epsilon})
	\end{equation}

	if for all $(\dhtilde,\dmtilde)\in \mathbb{R}^{K-1}\times\mathbb{R}^{N-1}$, there exists $\ell$, $k \in [L]$ such that
	\begin{align}
		&\sign((R_{\hhatnat}\bl)^\intercal \ve_1)(\widetilde{R_{\hhatnat}\bl})^\intercal \dhtilde\leq0 \text{ and } \sign((R_{\hat\mo}\cl)^\intercal \ve_1)(\widetilde{R_{\hat\mo}\cl})^\intercal \dmtilde\leq 0   \label{cond1alt}\\ 
		&\sign((R_{\hat\ho}\vb_k)^\intercal \ve_1)(\widetilde{R_{\hat\ho}\vb_k})^\intercal\dhtilde\geq 0 \text{ and } \sign((R_{\hat\mo}\vc_k)^\intercal \ve_1)(\widetilde{R_{\hat\mo}\vc_k})^\intercal \dmtilde\leq 0. \label{cond2alt}
	\end{align}	
	By Lemma \ref{sphere covering}, there exists $\ell, \ k \in [L]$ that satisfy \eqref{cond1alt} and \eqref{cond2alt}, respectively, with probability at least $1-e^{-cL}$ if $L\geq C(K+N)$.
\end{proof}

We now present a proof of Theorem \ref{thm:noise}. {\change In Theorem \ref{thm:noise}, the noise $\xi_\l \in [-1,1]$ which is in contrast to $\xi_\l \in [-1, 0]$ in Lemma \ref{good noise arbitrary}. The key idea is measurements with noise that satisfy $\xi_\l \geq -1$ can be converted to measurements with noise in the interval $[-1,1]$. In order to see this, let 
\begin{align}
	\bar{s} &= \max_{\l \in [L]} \frac{\yl}{\hat{\yl}} = 1+\max_{\l \in [L]}\xi_{\l},\\
	s &= \max\{\bar{s},1\}\label{shift}\leq 1+\|\vxi\|_\infty,\\
	 \eta_\ell &= \frac{1}{s}(1-s+\xi_\l). \label{eta}
\end{align} 
We then consider the measurements $\yl = s\hat{\yl}(1+\eta_\l)$ for $\l \in [L]$. Because $s\hat{\yl}(1+\eta_\l) = \hat{\yl}(1+\xi_\l)$, the noisy measurements are the same, however the noise may be different.
}
\begin{proof}[{\bf Proof of Theorem \ref{thm:noise}}]
	As the noise of measurements $y_\l = \hat{y_\l}(1 + \xi_\l)$ may not be one-sided as in \eqref{one-sided}, we consider equivalent measurements $y_\l = s\hat{y_\l}(1 + \eta_\l)$, where $s$ and $\eta_\l$ are as defined in \eqref{shift} and \eqref{eta}, respectively. This turns the BranchHull program \eqref{bh} into
	\begin{equation}\label{shift BH}
		\begin{aligned}
			\min_{\vh \in \mathbb{R}^K,\ \vm\in \mathbb{R}^N} \|\vh\|_{2}^{2}+\|\vm\|_{2}^{2} \hspace{5mm} \text{ s.t. }\hspace{5mm}  &\sign(y_{\ell})\blt \vh \clt \vm \geq |s\hat{y_{\ell}}(1+\eta_\l)|, \\
			\vspace{-10mm} & s_{\ell}\cdot\blt \vh \geq 0,\ \ell\in [L].
	\end{aligned}	
	\end{equation}	
	First, we note that for all $\l \in [L]$,
	\begin{align}
		\eta_\l & =  \frac{1}{s}\left(1+\xi_\l-s\right)\\
		&\leq \frac{1}{s}(s-s)\\
		& = 0,
	\end{align}
	where the first inequality holds because $1+\xi_\l \leq 1+\max_{\ell \in [L]} \xi_\l \leq s$. Second, we have $\eta_\l \geq -1$ for all $\l \in [L]$, which follows directly from $\xi_\l \geq -1$ for all $\l$. Thus, the noise $\veta$ is one-sided and by Lemma \ref{good noise arbitrary}, the minimizer $(\vh^*,\vm^*)$ of \eqref{shift BH} is unique and if $L\geq C(K+N)$, the minimizer satisfies
	\begin{equation}\label{shift guarantee}
		\left\|\vh^* - \sqrt{s}\ho\sqrt{\frac{\|\mo\|_2}{\|\ho\|_2}}\right\|_2^2+\left\|\vm^* - \sqrt{s}\mo\sqrt{\frac{\|\ho\|_2}{\|\mo\|_2}}\right\|_2^2\leq 4(1-\sqrt{1-\delta})s\|\ho\|_2\|\mo\|_2	
	\end{equation}

	with probability at least $1-e^{-cL}	$. In \eqref{shift guarantee},
	\begin{align}
		\delta &\equiv \|\veta\|_\infty \nonumber\\ 
		& = -\min_{\l \in [L]}\left(\frac{1}{s}-1+\frac{\xi_{\l}}{s}\right)\nonumber\\
			& = -\frac{1}{s}(1+\min_{\l \in [L]}\xi_\l)+1\nonumber\\
			& \leq -\frac{1 - \epsilon}{s}+1, \label{delta bound}
	\end{align}
	 where the first equality holds because $\epsilon =\|\vxi\|_\infty \geq -\min_{\l \in[L]}\xi_\l$. Let $(\hc,\mc) = \left(\ho\sqrt{\frac{\|\mo\|_2}{\|\ho\|_2}},\mo\sqrt{\frac{\|\ho\|_2}{\|\mo\|_2}}\right)$. We now compute
	\begin{align}
		& \left(\left\|\vh^* -\hc\right\|_2^2+\left\|\vm^* - \mc\right\|_2^2\right)^{\frac{1}{2}} \notag\\
		= & \left(\left\|\vh^* - \sqrt{s}\hc+\sqrt{s}\hc-\hc\right\|_2^2+\left\|\vm^* - \sqrt{s}\mc+\sqrt{s}\mc-\mc\right\|_2^2\right)^{\frac{1}{2}} \notag\\
		\leq & \left(\left\|\vh^* - \sqrt{s}\hc\right\|_{2}^2+\left\|\vm^* - \sqrt{s}\mc\right\|_2^2\right)^{\frac{1}{2}}+\left(\left\|\sqrt{s}\hc-\hc\right\|_{2}^2+\left\|\sqrt{s}\mc-\mc\right\|_2^2\right)^{\frac{1}{2}}\label{triangle 1}\\
		\leq & \left(4(1-\sqrt{1-\delta})s\|\ho\|_2\|\mo\|_2\right)^{\frac{1}{2}}+ (\sqrt{s}-1)\left(\|\hc\|_{2}^2+\|\mc\|_{2}^2\right)^{\frac{1}{2}}\label{apply lemma}\\
		\leq & {\change\left(2\left(s\left(1-\sqrt{\frac{1-\epsilon}{s}}\right)\right)^{\frac{1}{2}}+ \sqrt{2}\left(\sqrt{s}-1\right)\right)\sqrt{\|\ho\|_{2}\|\mo\|_{2}}}\label{bound d}\\
		\leq & {\change\left(2\left(1+\epsilon-\sqrt{1-\epsilon^2}\right)^{\frac{1}{2}}+ \sqrt{2}\left(\sqrt{1+\epsilon}-1\right)\right)\sqrt{\|\ho\|_{2}\|\mo\|_{2}}} \label{bound s}\\
		\leq & \left(2\sqrt{2\epsilon}+\frac{\sqrt{2}\epsilon}{2}\right)\sqrt{\|\ho\|_{2}\|\mo\|_{2}}\label{approx e 1}\\
		\leq & 4\sqrt{\epsilon}\sqrt{\|\ho\|_{2}\|\mo\|_{2}},\label{approx e 2}
	\end{align}
where \eqref{triangle 1} holds because of triangle inequality, \eqref{apply lemma} holds because of \eqref{shift guarantee}, \eqref{bound d} holds because of \eqref{delta bound} and $\|\hc\|_{2}=\|\mc\|_{2} = \sqrt{\|\ho\|_2\|\mo\|_2}$, \eqref{bound s} holds because of \eqref{shift} and \eqref{approx e 1} and \eqref{approx e 2} holds because for all $\epsilon \in [0,1]$, we have {\change $1+\epsilon -\sqrt{1-\epsilon^2}\leq 2\epsilon$, $\sqrt{1+\epsilon}-1\leq \frac{\epsilon}{2}$} and $\epsilon \leq \sqrt{\epsilon}$.
\end{proof}

\section{Numerical Results}\label{simulation}
In this section, we provide two numerical studies on synthetic data. The first study numerically verifies Theorem \ref{thm:sd} and the second study shows that the BranchHull program \eqref{bh} is robust to small dense noise.  {\change For both simulations, we used an interior point solver available in Matlab to solve the corresponding BranchHull program.}

For the first simulation, consider the following measurements: fix $N \in \{10,20,\dots, 150\}$, $L \in \{10,70,\dots,850\}$ and let $K = N$. Let the target signal $(\ho, \mo) \in \mathbb{R}^{K}\times \mathbb{R}^{N}$ be such that $\ho = e_1$ and $\mo = e_1$. Let $\mB \in \mathbb{R}^{L\times K}$ and $\mC \in \mathbb{R}^{L\times N}$ such that $B_{ij}\sim \mathcal{N}(0,1)$ and $C_{ij}\sim \mathcal{N}(0,1)$. Lastly, let $y_\l = \mB\ho\circ \mC\mo$ and $s = \text{sign}(\mB\ho)$.

Figure \ref{phase_plot_noiseless} shows the fraction of successful recoveries from 10 independent trials for the bilinear inverse problem \eqref{signal-recovery-problem} from data as described above. Black squares correspond to no successful recovery and white squares correspond to 100\% successful recovery. Let $(\vh^*,\vm^*)$ be the output of \eqref{bh}. For each trial, we say \eqref{bh} successfully recovers the target signal if  $\|(\vh^*,\vm^*)-(\ve_1,\ve_1)\|_{2} <10^{-5}$. The area to the left of the line corresponds to the oversampling  required for successful recovery stated in Theorem \ref{thm:sd}. The figures shows the linear relationship between number of measurements $L$ and size of target signals $K+N$ for successful recovery.
\begin{figure}[H]
		\centering
		\includegraphics[scale  = 0.5]{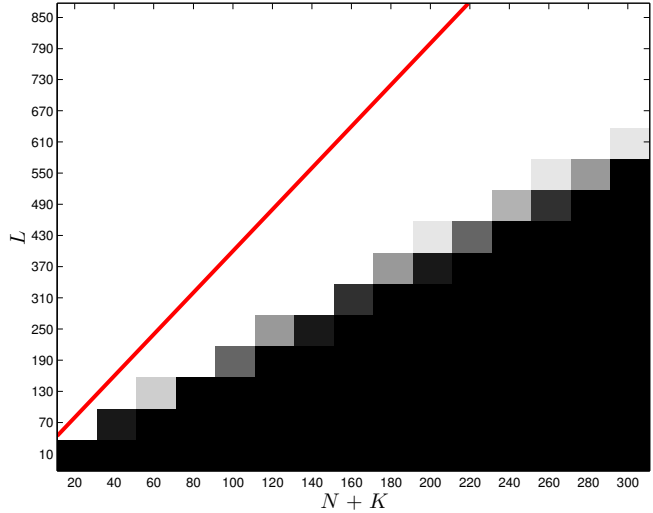}
		\captionsetup{oneside,margin={0em,0em}}
		\caption{ The empirical recovery probability from synthetic data with total measurements $L$ as a function of size of the target signals $K+N$. The shades of black and white represents the fraction of successful simulation. White blocks correspond to successful recovery and black blocks correspond to unsuccessful recovery. Each block corresponds to the average from 10 independent trials. The area to the left of the line satisfies {\change $L >2(K+N),$} which is the theoretical successful recovery bound stated in Theorem \ref{thm:sd}.}
		\label{phase_plot_noiseless}
\end{figure}

For the noisy simulation, consider the following measurements: fix $N = K = 20$ and $L\in \{10,20,\dots,200\}$. Let the target signal $(\ho, \mo) \in \mathbb{R}^{K}\times \mathbb{R}^{N}$ be such that $\ho \sim \mathcal{N}(0,1)$ and $\mo \sim \mathcal{N}(0,1)$. Let $\mB \in \mathbb{R}^{L\times K}$ and $\mC \in \mathbb{R}^{L\times N}$ be such that $B_{ij}\sim \mathcal{N}(0,1)$ and $C_{ij}\sim \mathcal{N}(0,1)$. Fix $\alpha \in \{0,0.1,\dots, 1\}$ and let $\xi \in \mathbb{R}^{L}$ such that $\xi_\l \sim \text{Uniform}([-\alpha,\alpha])$. Lastly, let $y_\l = \mB\ho\circ \mC\mo\circ ({\bf 1} + \xi)$ and $s = \text{sign}(\mB\ho)$. 

Figure \ref{vary_noise_level} shows the maximum relative error from 10 independent trials for the bilinear inverse problem \eqref{signal-recovery-problem} from data as described above. Each curve corresponds to different noise level $\alpha${\change, which controls the size of noise level $\epsilon$ defined in \eqref{epsilon}}. The plot shows the effect of different levels of noise on the maximum relative error as a function of the sampling ratio $\frac{L}{K+N}$.  {\change Empirically, sampling ratio of about 2.5 is sufficient for stable estimation of the target signal. Additionally, the spacing between the piecewise lines for large sampling ratio is uniform in Figure \ref{vary_noise_level}. This suggests that the relationship between recovery error and the noise level $\epsilon$ is linear, which is in contrast to Theorem \ref{thm:noise} which shows that the recovery error depends as $\sqrt{\epsilon}$ for small noise.
} 
\begin{figure}[H]
		\centering
		\includegraphics[scale  = 0.5]{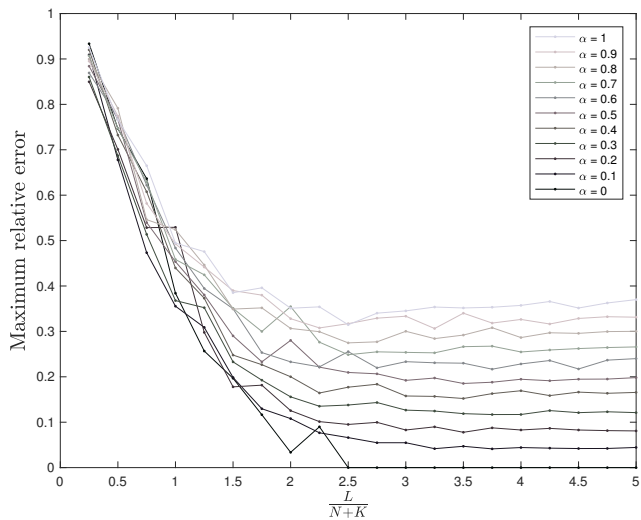}
		\captionsetup{oneside,margin={0em,0em}}
		\caption{The empirical recovery error from synthetic data as a function of sampling ratio $\frac{L}{K+N}$. The size of the signals, $N$ and $K$, are fixed to 20. Different piecewise line corresponds to different noise level $\alpha$.}
		\label{vary_noise_level}
\end{figure}

\subsubsection*{Acknowledgments}
PH acknowledges funding by the grant NSF DMS-1464525.
\bibliography{Bibliography1}

\begin{thebibliography}{30}
\providecommand{\natexlab}[1]{#1}
\providecommand{\url}[1]{\texttt{#1}}
\expandafter\ifx\csname urlstyle\endcsname\relax
  \providecommand{\doi}[1]{doi: #1}\else
  \providecommand{\doi}{doi: \begingroup \urlstyle{rm}\Url}\fi

\bibitem[Ahmed et~al.(2014)Ahmed, Recht, and Romberg]{ahmed2012blind}
Ali Ahmed, Benjamin Recht, and Justin Romberg.
\newblock Blind deconvolution using convex programming.
\newblock \emph{{IEEE} Trans. Inform. Theory}, 60\penalty0 (3):\penalty0
  1711--1732, 2014.

\bibitem[Stockham et~al.(1975)Stockham, Cannon, and
  Ingebretsen]{stockham1975blind}
Thomas~G Stockham, Thomas~M Cannon, and Robert~B Ingebretsen.
\newblock Blind deconvolution through digital signal processing.
\newblock \emph{Proceedings of the IEEE}, 63\penalty0 (4):\penalty0 678--692,
  1975.

\bibitem[Kundur and Hatzinakos(1996)]{kundur1996blind}
Deepa Kundur and Dimitrios Hatzinakos.
\newblock Blind image deconvolution.
\newblock \emph{IEEE signal processing magazine}, 13\penalty0 (3):\penalty0
  43--64, 1996.

\bibitem[Aghasi et~al.(2016)Aghasi, Heshmat, Redo-Sanchez, Romberg, and
  Raskar]{aghasi2016sweep}
Alireza Aghasi, Barmak Heshmat, Albert Redo-Sanchez, Justin Romberg, and Ramesh
  Raskar.
\newblock Sweep distortion removal from terahertz images via blind
  demodulation.
\newblock \emph{Optica}, 3\penalty0 (7):\penalty0 754--762, 2016.

\bibitem[Fienup(1982)]{fienup1982phase}
James~R Fienup.
\newblock Phase retrieval algorithms: a comparison.
\newblock \emph{Applied optics}, 21\penalty0 (15):\penalty0 2758--2769, 1982.

\bibitem[Tosic and Frossard(2011)]{tosic2011dictionary}
Ivana Tosic and Pascal Frossard.
\newblock Dictionary learning.
\newblock \emph{IEEE Signal Processing Magazine}, 28\penalty0 (2):\penalty0
  27--38, 2011.

\bibitem[Hoyer(2004)]{hoyer2004non}
Patrik~O Hoyer.
\newblock Non-negative matrix factorization with sparseness constraints.
\newblock \emph{Journal of machine learning research}, 5\penalty0
  (Nov):\penalty0 1457--1469, 2004.

\bibitem[Lee and Seung(2001)]{lee2001algorithms}
Daniel~D Lee and H~Sebastian Seung.
\newblock Algorithms for non-negative matrix factorization.
\newblock In \emph{Advances in neural information processing systems}, pages
  556--562, 2001.

\bibitem[Ling and Strohmer(2015)]{ling2015self}
Shuyang Ling and Thomas Strohmer.
\newblock Self-calibration and biconvex compressive sensing.
\newblock \emph{Inverse Problems}, 31\penalty0 (11):\penalty0 115002, 2015.

\bibitem[Castro(2015)]{castro2015tightening}
Pedro~M Castro.
\newblock Tightening piecewise mccormick relaxations for bilinear problems.
\newblock \emph{Computers \& Chemical Engineering}, 72:\penalty0 300--311,
  2015.

\bibitem[Chen et~al.(2006)Chen, Yin, Zhou, Comaniciu, and
  Huang]{Chen2006variation}
T.~Chen, Wotao Yin, Xiang~Sean Zhou, D.~Comaniciu, and T.~S. Huang.
\newblock Total variation models for variable lighting face recognition.
\newblock \emph{IEEE Transactions on Pattern Analysis and Machine
  Intelligence}, 28\penalty0 (9):\penalty0 1519--1524, Sept 2006.
\newblock ISSN 0162-8828.
\newblock \doi{10.1109/TPAMI.2006.195}.

\bibitem[Ahmed and Demanet(2016)]{ahmed2016leveraging}
Ali Ahmed and Laurent Demanet.
\newblock Leveraging diversity and sparsity in blind deconvolution.
\newblock \emph{arXiv preprint arXiv:1610.06098}, 2016.

\bibitem[Netrapalli et~al.(2013)Netrapalli, Jain, and
  Sanghavi]{netrapalli2013phase}
Praneeth Netrapalli, Prateek Jain, and Sujay Sanghavi.
\newblock Phase retrieval using alternating minimization.
\newblock In \emph{Advances Neural Inform. Process. Syst.}, pages 2796--2804,
  2013.

\bibitem[Sun et~al.(2016)Sun, Qu, and Wright]{sun2016geometric}
Ju~Sun, Qing Qu, and John Wright.
\newblock A geometric analysis of phase retrieval.
\newblock In \emph{Information Theory (ISIT), 2016 IEEE International Symposium
  on}, pages 2379--2383. IEEE, 2016.

\bibitem[Cand\`es et~al.(2015)Cand\`es, Li, and
  Soltanolkotabi]{candes2014phase}
Emmanuel Cand\`es, Xiaodong Li, and Mahdi Soltanolkotabi.
\newblock Phase retrieval via wirtinger flow: Theory and algorithms.
\newblock \emph{IEEE Trans. Inform. Theory}, 61\penalty0 (4):\penalty0
  1985--2007, 2015.

\bibitem[Chen and Cand\`es(2015)]{CC15}
Yuxin Chen and Emmanuel Cand\`es.
\newblock Solving random quadratic systems of equations is nearly as easy as
  solving linear systems.
\newblock In \emph{Advances Neural Inform. Process. Syst.}, pages 739--747,
  2015.

\bibitem[Wang et~al.(2018)Wang, Giannakis, and Eldar]{wang2016solving}
G.~Wang, G.~B. Giannakis, and Y.~C. Eldar.
\newblock Solving systems of random quadratic equations via truncated amplitude
  flow.
\newblock \emph{IEEE Transactions on Information Theory}, 64\penalty0
  (2):\penalty0 773--794, Feb 2018.
\newblock ISSN 0018-9448.
\newblock \doi{10.1109/TIT.2017.2756858}.

\bibitem[Li et~al.(2016)Li, Ling, Strohmer, and Wei]{li2016rapid}
Xiaodong Li, Shuyang Ling, Thomas Strohmer, and Ke~Wei.
\newblock Rapid, robust, and reliable blind deconvolution via nonconvex
  optimization.
\newblock \emph{CoRR}, abs/1606.04933, 2016.

\bibitem[Tu et~al.(2016)Tu, Boczar, Simchowitz, Soltanolkotabi, and
  Recht]{tu2015low}
Stephen Tu, Ross Boczar, Max Simchowitz, Mahdi Soltanolkotabi, and Benjamin
  Recht.
\newblock Low-rank solutions of linear matrix equations via procrustes flow.
\newblock In Maria~Florina Balcan and Kilian~Q. Weinberger, editors,
  \emph{Proceedings of the 33rd International Conference on International
  Conference on Machine Learning - Volume 48}, volume~48 of \emph{Proceedings
  of Machine Learning Research}, pages 964--973. PMLR, 2016.

\bibitem[Bahmani and Rombger(2017)]{bahmani2016phase}
Sohail Bahmani and Justin Rombger.
\newblock {Phase Retrieval Meets Statistical Learning Theory: A Flexible Convex
  Relaxation}.
\newblock In Aarti Singh and Jerry Zhu, editors, \emph{Proceedings of the 20th
  International Conference on Artificial Intelligence and Statistics},
  volume~54 of \emph{Proceedings of Machine Learning Research}, pages 252--260.
  PMLR, 20-22 Apr 2017.

\bibitem[Goldstein and Studer(2018)]{goldstein2016phasemax}
T.~Goldstein and C.~Studer.
\newblock Phasemax: Convex phase retrieval via basis pursuit.
\newblock \emph{IEEE Transactions on Information Theory}, 64\penalty0
  (4):\penalty0 2675--2689, April 2018.
\newblock ISSN 0018-9448.
\newblock \doi{10.1109/TIT.2018.2800768}.

\bibitem[Hand and Voroninski(2016{\natexlab{a}})]{hand2016elementary}
Paul Hand and Vladislav Voroninski.
\newblock An elementary proof of convex phase retrieval in the natural
  parameter space via the linear program phasemax.
\newblock \emph{arXiv preprint arXiv:1611.03935}, 2016{\natexlab{a}}.

\bibitem[McCormick(1976)]{mccormick1976computability}
Garth~P McCormick.
\newblock Computability of global solutions to factorable nonconvex programs:
  Part i -- convex underestimating problems.
\newblock \emph{Mathematical programming}, 10\penalty0 (1):\penalty0 147--175,
  1976.

\bibitem[Quesada and Grossmann(1995)]{quesada1995global}
Ignacio Quesada and Ignacio~E Grossmann.
\newblock A global optimization algorithm for linear fractional and bilinear
  programs.
\newblock \emph{Journal of Global Optimization}, 6\penalty0 (1):\penalty0
  39--76, 1995.

\bibitem[Horst and Tuy(2013)]{horst2013global}
Reiner Horst and Hoang Tuy.
\newblock \emph{Global optimization: Deterministic approaches}.
\newblock Springer Science \& Business Media, 2013.

\bibitem[Hand and Voroninski(2016{\natexlab{b}})]{hand2016corruption}
Paul Hand and Vladislav Voroninski.
\newblock Corruption robust phase retrieval via linear programming.
\newblock \emph{arXiv preprint arXiv:1612.03547}, 2016{\natexlab{b}}.

\bibitem[Wendel(1962)]{wendel1962problem}
James~G Wendel.
\newblock A problem in geometric probability.
\newblock \emph{Math. Scand}, 11:\penalty0 109--111, 1962.

\bibitem[Schl{\"a}fli(1953)]{schlafli1953gesammelte}
Ludwig Schl{\"a}fli.
\newblock \emph{Gesammelte mathematische Abhandlungen: 1814-1895. 2}.
\newblock Birkh{\"a}user, 1953.

\bibitem[Wasserman(2013)]{wasserman2013all}
Larry Wasserman.
\newblock \emph{All of statistics: a concise course in statistical inference}.
\newblock Springer Science \& Business Media, 2013.

\bibitem[Vershynin(2012)]{vershynin10in}
R.~Vershynin.
\newblock Introduction to the non-asymptotic analysis of random matrices.
\newblock In Yonina~C Eldar and Gitta Kutyniok, editors, \emph{Compressed
  sensing: theory and applications}. Cambridge University Press, 2012.

\end{thebibliography}
\end{document}